\documentclass[lettersize,journal]{IEEEtran}
\usepackage{amsmath,amsfonts}
\usepackage{algorithmic}
\usepackage{algorithm}
\usepackage{array}
\usepackage[caption=false,font=normalsize,labelfont=sf,textfont=sf]{subfig}
\usepackage{textcomp}
\usepackage{stfloats}
\usepackage{url}
\usepackage{verbatim}
\usepackage{graphicx}
\usepackage{cite}
\usepackage{xcolor}     
\usepackage{multirow}   
\usepackage{tabularx}   
\usepackage{makecell}   
\usepackage{pdfpages}   
\usepackage{setspace}   
\usepackage{booktabs}   

\usepackage{amsthm}
\theoremstyle{definition}

\newtheorem{theorem}{\normalfont\bfseries Theorem}

\newtheorem{definition}{\normalfont\bfseries Definition}

\newtheorem{corollary}{\normalfont\bfseries Corollary}

\newtheorem{remark}{\normalfont\bfseries Remark}

\newcommand{\kd}{k_{\rm d}}
\newcommand{\ks}{k_{\rm s}}
\newcommand{\vmax}{v_{\rm max}}
\newcommand{\Dst}{D_{\rm st}}
\newcommand{\Dsf}{D_{\rm sf}}
\newcommand{\ksf}{\kappa_{\rm{sf}}}
\newcommand{\GaR}{\Gamma_{\rm R}(\omega)}
\newcommand{\GaI}{\Gamma_{\rm I}(\omega)}

\title{Safe and Stable Connected Cruise Control for Connected Automated Vehicles with Response Lag

\author{Yuchen Chen, G{\'{a}}bor Orosz, and Tamas G. Molnar}
\thanks{Y. Chen is with the Department of Mechanical Engineering, University of Michigan, Ann Arbor, MI 48109, USA,
{\tt\small ethanch@umich.edu}.}%
\thanks{G. Orosz is with the Department of Mechanical Engineering and with the Department of Civil and Environmental Engineering, University of Michigan, Ann Arbor, MI 48109, USA,
{\tt\small orosz@umich.edu}.}%
\thanks{T. G. Molnar is with the Department of Mechanical Engineering, Wichita State University, Wichita, KS 67260, USA,
{\tt\small tamas.molnar@wichita.edu}.}%
}

\begin{document}
\maketitle

\begin{abstract}
Controlling connected automated vehicles (CAVs) via vehicle-to-everything (V2X) connectivity holds significant promise for improving fuel economy and traffic efficiency.
However, to deploy CAVs and reap their benefits, their controllers must guarantee their safety.
In this paper, we apply control barrier function (CBF) theory to investigate the safety of CAVs implementing connected cruise control (CCC).
Specifically, we study how stability, connection architecture, and the CAV's response time impact the safety of CCC.
Through safety and stability analyses, we derive stable and safe choices of control gains, and show that safe CAV operation requires plant and head-to-tail string stability in most cases.
Furthermore, the reaction time of vehicles—which is represented as a first-order lag—has a detrimental effect on safety. 
We determine the critical value of this lag, above which safe CCC gains do not exist.
To guarantee safety even with lag while preserving the benefits of CCC, we synthesize {\em safety-critical CCC} using CBFs.
With the proposed safety-critical CCC, the CAV can leverage information from connected vehicles farther ahead to improve its safety.
We evaluate this controller by numerical simulation using real traffic data.
\end{abstract}


\section{Introduction}
\label{sec:intro}

\IEEEPARstart{T}{he} promise of automated vehicles (AVs) to greatly improve safety, passenger comfort, fuel efficiency, and traffic efficiency has led to a surge of research and development on AVs in recent years.
As a result, longitudinal controller designs for AVs such as {\em adaptive cruise control} (ACC) have been studied intensively\cite{Gunter2021arecommercially,Bekiaris-Liberis2018} and deployed extensively.
Vehicle-to-everything (V2X) connectivity, which enables connected automated vehicles (CAVs) to exchange information with other road users, can further enhance the performance of AVs.
With more information about the surrounding traffic, strategies like {\em cooperative adaptive cruise control }(CACC)~\cite{wang2018review_CACC} enable platoons of CAVs to drive cooperatively to save fuel \cite{Turri2015CACCfuel} and achieve string stability \cite{Van2019CACCstability}.
Additionally, {\em connected cruise control} (CCC)~\cite{zhang2016motif}, which controls the motion of a single CAV based on the information shared by the preceding connected (but not necessarily automated) vehicles, outperforms ACC in experiments \cite{jin2018experimental} and also shows other significant benefits, such as reduced traffic congestion \cite{Qin2018CCCtraffic} and improved fuel economy \cite{Shen2023Energy}.
Alternatively, {\em leading cruise control} (LCC)~\cite{wang2022LCC} and {\em connected traffic control} (CTC)~\cite{TRSC2024} utilize information from following vehicles to smooth the traffic flow.

\begin{figure}[t]
    \centering
    \includegraphics{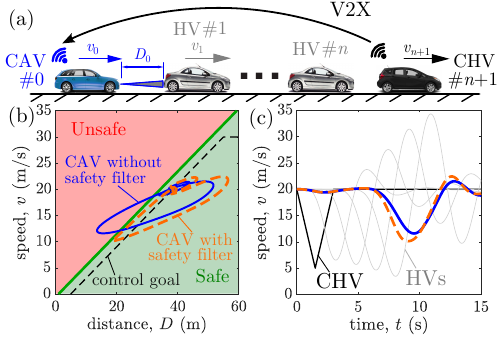}
    \vspace{-3mm}
    \caption{
    (a) A connected automated vehicle (CAV) responds to the preceding human-driven vehicles (HVs) and to the head connected human-driven vehicle (CHV) using vehicle-to-everything (V2X) connectivity and connected cruise control (CCC). (b-c) Behavior of the CCC with unsafe control parameter choice (blue), and the proposed safety-critical CCC that minimally modifies the above CCC to ensure safety using the safety filter (orange).}
    \label{fig:3 veh setup}
    \vspace{-3mm}
\end{figure}

Safety is the top priority when deploying AVs and CAVs.
Recently, the safety of mixed traffic (including automated and human-driven vehicles) has been studied, considering reinforcement learning~\cite{cheng2021enhancing} and head-to-tail string stability analysis~\cite{liu2023structural}.
Furthermore, safety-critical controllers have been developed for AVs and CAVs with approaches including reachability analysis~\cite{alam2014guaranteeing}, formal methods~\cite{nilsson2015correct}, reinforcement learning~\cite{Li2020SafeReinforcementLearning}, and model predictive control~\cite{massera2017safe, shen2024energy}.
Additionally, control barrier functions (CBFs) have gained popularity for safety-critical control.
Some applications of CBFs include adaptive cruise control~\cite{ames2014control, Waqas2022ACCCBF}, obstacle avoidance with AVs~\cite{chen2018obstacle}, lane changing~\cite{hu2023safety}, roundabout crossing~\cite{abduljabbar2021cbfbased}, merging scenarios~\cite{hao2023merge}, and safe traffic control by CAVs~\cite{zhao2023safetycritical}, with experimental validations~\cite{gunter2022experimental}.
Notably, CBFs were used in~\cite{He2018} to establish safe CCC, and a thorough analysis was then conducted in~\cite{Tamas2023CDC}.
A robust, safe CCC design was proposed in~\cite{alan2023control} with experimental validation on a heavy-duty truck.
However, these studies predominantly focus on scenarios where the CAV responds solely to the immediate preceding vehicle, while the relationship between connectivity to multiple vehicles ahead and safety is yet to be explored.
In addition, while safety-critical control for time delay systems with state delay has been addressed in previous studies such as~\cite{kiss2023safetydelay, ren2021razumikhin, ren2022razumikhin}, input delays are still not taken into account for safety-critical controller synthesis.

To address these gaps, our recent work~\cite{chen2024CCCsafety} investigated the safety of CCC in connected vehicle networks, where CAVs respond to the connected human-driven vehicles (CHVs) farther ahead via V2X connectivity, as illustrated in Fig.~\ref{fig:3 veh setup}(a).
Using safety charts, we provide a guideline to select safe control gains for existing CCC designs.
Furthermore, we also proposed a safety-critical CCC method, which allows existing, efficient CCC algorithms to operate when it is safe and only minimally intervenes when danger appears to ensure safety; see Fig.~\ref{fig:3 veh setup}(b-c).
Compared to our previous work in~\cite{chen2024CCCsafety}, this paper contains the following contributions:
\begin{itemize}
    \item We capture the response time of CAVs as a first-order lag applied to the control input and analyze how it affects the safety of existing CCC designs.
    Via safety charts, we show that the lag negatively impacts safety. 
    We quantify a critical value of the lag for which safe choices of control gains cease to exist.
    \item We extend the safety analysis to cases with the general connectivity structure, where the CAV responds to multiple connected vehicles ahead of it, and potentially also utilizes information about their acceleration. 
    \item Apart from studying safety, we also provide a detailed stability analysis for the general CCC designs mentioned above. 
    Moreover, we study the interplay of stability and safety by comparing stability charts with safety charts, and we demonstrate that in most cases, existing CCC designs can achieve safety when the vehicle chain is both plant and head-to-tail string stable.
    \item To overcome the limitations of existing CCC methods, we propose a safety-critical CCC algorithm based on CBF theory that takes the lag into account and guarantees safety.
    This is achieved by minimally modifying existing CCC laws that have high performance (i.e., provide plant and head-to-tail string stability) but lack safety guarantees.
    \item We evaluate the performance of safety-critical CCC via numerical simulations, while utilizing real traffic data for the motions of human-driven vehicles. 
    Through simulation, we show that the proposed controller can leverage connectivity to improve safety, and it guarantees safety regardless of the lag and connectivity architecture while preserving the high performance of CCC.
\end{itemize}

The rest of this paper is organized as follows.
In Section~\ref{sec:ccc}, we introduce longitudinal vehicle dynamics with first-order lag and the framework of CCC. 
In Section~\ref{sec:stability}, we derive conditions for the plant stability and head-to-tail string stability of CCC. 
In Section~\ref{sec:CBF}, we provide background on CBFs. 
In Section~\ref{sec:safccc}, we establish safety conditions for existing CCC laws, and visualize them via safety charts.
Next, we propose {\em safety-critical CCC} with lag, and demonstrate its behavior using numerical simulations. 
Finally, we conclude our results in Section~\ref{sec:concl}.

\section{Connected Cruise Control}
\label{sec:ccc}

Consider a setup where a connected automated vehicle (CAV) is controlled to follow a chain of human-driven vehicles (HVs) on a single lane while responding to connected human-driven vehicles (CHVs) farther ahead. 
The CAV uses on-board sensors to evaluate its own speed $v_0$, the preceding HV's speed $v_1$, and the distance $D_0$ between them. Additionally, the CAV also receives the CHVs' speeds $v_k$ through vehicle-to-everything (V2X) connectivity. An illustration of this setup is depicted in Fig.~\ref{fig:3 veh setup}(a), where the CAV is only linked with one CHV that is ${n+1}$ vehicles ahead.

We represent the motion of the HVs using the dynamics:
\begin{align}
\begin{split}
    \dot{D}_{i}(t) & = v_{i+1}(t) - v_{i}(t), 
    \\
    \dot{v}_{i}(t) & = a_{i}(t),
    \\
    a_i(t) &= u_{i}(t-\tau),
\end{split}
\label{eq:HV system}
\end{align}
for $i \in \{1,...,n\}$, where $a_i$ is the actual acceleration, $u_{i}$ is the desired one, and $\tau$ captures driver reaction time and powertrain delays.
We use the optimal velocity model (OVM) to characterize the HVs' driving behavior during the car-following task~\cite{bando1998analysis}: 
\begin{equation}
    u_{i} = A_{i} \big( V_i(D_{i}) - v_{i} \big) + B_{i} \big(v_{i+1} - v_{i} \big),
\label{eq:HVcon}
\end{equation}
for ${i \in \{1,...,n\}}$. 
In this model, each HV's desired acceleration is determined by the speed difference with gain $B_i$ and the distance with gain $A_i$ and the range policy:
\begin{equation}
    V_i(D) = \min\{\kappa_i(D-\Dst),\vmax\},
\label{eq:HV range policy}
\end{equation}
for ${i \in \{1,...,n\}}$. 
This policy outlines a target velocity determined by the distance to the vehicle ahead, which is zero at the standstill distance $\Dst$ and grows linearly with gradient $\kappa_i$ up to the speed limit $\vmax$.

To describe the CAV's dynamics, we incorporate its response time $\xi$ as a first-order lag:
\begin{align}
\begin{split}
    \dot{D}_0(t) & = v_{1}(t) - v_{0}(t), 
    \\
    \dot{v}_0(t) & = a_{0}(t), 
    \\
    \dot{a}_0(t) & = \frac{1}{\xi} \big(u_{0}(t) - a_{0}(t)\big).
\end{split}
\label{eq:CAV system w/ lag}
\end{align}
The CAV executes the acceleration command $u_{0}$, 
while its actual acceleration $a_0$ has lag $\xi$.
Note that when ${\xi \rightarrow 0}$, ${u_{0}(t) \equiv a_{0}(t)}$ and system (\ref{eq:CAV system w/ lag}) can be simplified to the first two equations.
For simplicity, we neglect the rolling and air resistance of the CAV, which could potentially reduce speed and improve safety.
We also omit input constraints for the CAV's dynamics, since well-tuned controllers can successfully prevent the input from exceeding the physical limit, as will be shown in the following content.
Notably, safety-critical control with input constraints is an open problem to be explored.
This is addressed in the context of control barrier functions, for example, in~\cite{xiao2022sufficient, chen2021backup, Ames2021icbf}.

\begin{figure}[t]
    \centering
    \includegraphics{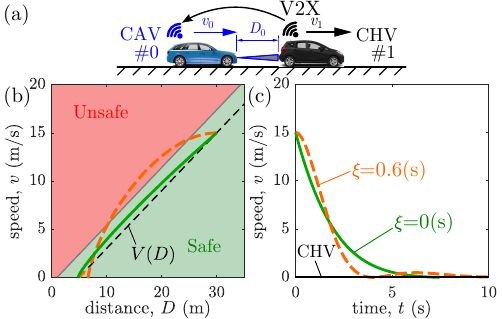}
    \vspace{-3mm}
    \caption{Simulations of system (\ref{eq:CAV system w/ lag}) for $\xi=0\, \mathrm{s}$ (green) and $\xi=0.6\, \mathrm{s}$ (orange) using CCC (\ref{eq:CCC general}) with $n=1$, where a CAV performs an emergency brake to avoid crashing with the stopped CHV.}
    \label{fig:2veh nominal CCC w/ lag}
    \vspace{-3mm}
\end{figure}

\setlength{\tabcolsep}{1pt}
\begin{table}
\caption {Parameters of the numerical case studies}
\vspace{-3mm}
\label{tab:parameters}
\begin{center}
\small
\begin{tabularx}{\columnwidth}{ccccc}
    \toprule
    Vehicle & Variable & Symbol & Value & Unit \\
    \midrule
    \multirow{3}{*}{All}
    & speed limit &$\vmax$ & $30$ & $\mathrm{m/s}$ \\
    & standstill distance &$D_{\mathrm{st}}$ & $5$ & $\mathrm{m}$ \\
    & acceleration limit & $(a_{\rm min},a_{\rm max})$ & $(7,3)$ & $\mathrm{m/s^{2}}$ \\
    \midrule
    \multirow{3}{*}{HV}
    & delay & $\tau$& $0.9$ & $\mathrm{s}$ \\
    & range policy gradient &$\kappa_{\rm h}$ & $0.6$ & $\mathrm{1/s}$ \\ 
    & driver parameters &$(A_{\rm h},B_{\rm h})$ & $(0.1,0.6)$ & $\mathrm{1/s}$ \\
    \midrule
    \multirow{8}{*}{CAV}
    & range policy gradient & $\kappa$ & $0.6$ & $\mathrm{1/s}$ \\ 
    & safe gains (point P) & $(A,B_1,B_{n+1})$ & $(0.6,0.53,0.03)$ & $\mathrm{1/s}$ \\ 
    & unsafe gains (point Q) & $(A,B_1,B_{n+1})$ & $(0.6,0.53,0.5)$ & $\mathrm{1/s}$ \\ 
    & safe distance & $\Dsf$& $1$ & $\mathrm{m}$ \\
    & inverse time headway &$\kappa_{\rm{sf}}$ & $0.6$ & $\mathrm{1/s}$ \\ 
    & speed difference limit &$\bar{v}$ & $15$ & $\mathrm{m/s}$ \\ 
    & CBF parameter &$\gamma$ & $1$ & $\mathrm{1/s}$ \\ 
    & \makecell{eCBF parameter} &$\gamma_{\rm e}$ & $1$ & $\mathrm{1/s}$ \\ 
    \midrule
    \multirow{1}{*}{CHV}
    & speed perturbation &$v_{\rm pert}$ & $15$ & $\mathrm{m/s}$ \\ 
    \bottomrule
\end{tabularx}
\normalsize
\end{center}
\vspace{-5mm}
\end{table}

For car-following purposes, we implement a desired controller for the CAV, ${u_{0}=\kd(x)}$, which responds to the state ${x= \begin{bmatrix} D_0 & v_{0} & a_{0} & v_{1} \end{bmatrix}^{\top}}$.
Specifically, we select the {\em connected cruise control} (CCC) strategy, which was first presented in~\cite{zhang2016motif} and validated through experiments in~\cite{jin2018experimental}.
Here we follow~\cite{zhang2016motif}, but with slightly changed notation that fits our purposes:
\begin{align}
\begin{split}
    \kd(x) =  A \big( V(D_0) - v_0 \big) + B_1 \big( W(v_{1}) - v_0 \big) &
    \\
    + \sum_{k \in \Phi} B_k \big( W(v_{k}) - v_0 \big)&.
\end{split}
\label{eq:CCC general}
\end{align}
CCC responds to the distance from and velocity of the vehicle ahead with gains $A$ and $B_1$, respectively, and also to the speeds of CHVs with gains $B_k$ ${(k > 1)}$. 
Here, $v_k$ denotes the speed of the CHV, which is $k$ vehicles ahead, and $\Phi$ denotes the set of indices corresponding to the CHVs whose velocity are accessible to the CAV via connectivity. 
For instance, in Fig.~\ref{fig:3 veh setup}(a), the CAV only connects to CHV${\ \#n\rm{+1}}$, thus ${\Phi=\{n+1\}}$. In this case, (\ref{eq:CCC general}) becomes:
\begin{align}
\begin{split}
    \kd(x) = A \big( V(D_0) - v_0 \big) + B_1 \big( W(v_{1}) - v_0 \big) & \\
    + B_{n+1} \big( W(v_{n+1}) - v_0 \big) &.
\end{split}
\label{eq:CCC 3veh}
\end{align}
When the CAV only responds to the preceding vehicle (${\Phi=\emptyset}$), CCC~(\ref{eq:CCC general}) only contains distance and speed difference terms with gains $A$ and $B_1$ like the OVM~(\ref{eq:HVcon}).
CAV implements the range policy: 
\begin{align}
\begin{split}
    V(D) & = \min\{\kappa(D-\Dst),\vmax\},
\end{split}
\label{eq:V}
\end{align}
and the speed policy:
\begin{align}
\begin{split}
    W(v) & = \min\{v,\vmax\},
\end{split}
\label{eq:W}
\end{align}
to ensure it does not exceed the speed limit.

Notice that CCC (\ref{eq:CCC general}) does not provide formal safety guarantees. 
This is illustrated in Fig.~\ref{fig:2veh nominal CCC w/ lag}, where CCC, given by~(\ref{eq:CAV system w/ lag},\ref{eq:CCC general}) with ${\xi=0\, \mathrm{s}, 0.6\, \mathrm{s}}$ and ${n=1}$, is simulated in a 2-vehicle scenario shown in panel (a).
We use $A=0.6$, $B_1=0.5$ for CCC (\ref{eq:CCC general}) and other parameters listed in Table~\ref{tab:parameters}.
The preceding CHV stops suddenly, and the CAV responds to it with a lag; see panel (c).
As shown by panel (b), the CAV closely follows the range policy $V(D)$ if the lag is small (green).
However, as the lag is increased (orange), the CAV maintains shorter distance and becomes unsafe.
Motivated by this, we will formally analyze how lag affects the safety of CAVs and calculate the critical lag above which safe choices of the gains $A$, $B_1$ and $B_k$ do not exist for CCC~(5).
Before the safety analysis, first we discuss stability analysis as preliminary.

\section{Stability Analysis}
\label{sec:stability}

Realizing stable behavior for the CAV is a fundamental requirement for CCC. 
In this section, we derive the plant and head-to-tail string stability conditions~\cite{zhang2016motif, xiao2008Stabilitylag} for the vehicle chain containing one CAV at the tail and multiple HVs and CHVs in front, see Fig.~\ref{fig:3 veh setup}(a).
Plant stability indicates that a vehicle is able to approach a constant-speed equilibrium in an asymptotically stable manner.
Head-to-tail string stability means that velocity fluctuations are attenuated as they propagate from the head to the tail vehicle along the chain. 
Plant and head-to-tail string stability can be evaluated by calculating the so-called head-to-tail transfer function \cite{zhang2016motif},
which is established in the Laplace domain for the corresponding linearized dynamics.

To derive the head-to-tail transfer function, first we linearize the systems~(\ref{eq:HV system},\ref{eq:HVcon}) and~(\ref{eq:CAV system w/ lag},\ref{eq:CCC general}) around the equilibrium:
\begin{align}
\begin{split}
        D_i(t) \equiv D_{i}^{\ast},\quad 
        v_i(t) \equiv v^{\ast},\quad 
        a_{i}(t) \equiv 0,
\end{split}
\label{eq:equilibrium}
\end{align}
for ${i \in \{0,...,n\}}$, where all vehicles drive at the uniform equilibrium speed $v^{\ast}$, while maintaining equilibrium distance $D_{i}^{\ast}$ given by ${v^{\ast} = V(D_{0}^{\ast}) = V_i(D_{i}^{\ast})}$ for ${i \in \{1,...,n\}}$. 

Then we define perturbations around the equilibrium: 
\begin{align}
        \Tilde{D}_i = D_i - D_{i}^{\ast},\quad 
        \Tilde{v}_i = v_i - v^{\ast},\quad  
        \Tilde{a}_i = a_i,
\label{eq:perturbations}
\end{align}
for ${i \in \{0,...,n\}}$, and consider state vectors with these perturbations as 
${x_0 = [\, \Tilde{D}_0 \, \ \Tilde{v}_0 \, \ \Tilde{a}_0 \,]^{\top}}$ 
for the CAV and 
${x_i = [\, \Tilde{D}_i \, \ \Tilde{v}_i \,]^{\top}}$ 
for the HVs ${i \in \{1,...,n\}}$. The speed fluctuations are obtained as:
\begin{align}
\begin{split}
        \Tilde{v}_0 &= \mathbf{c}_0 x_0,\quad \mathbf{c}_0 =\begin{bmatrix} 0 & 1 & 0 \end{bmatrix}, 
        \\
        \Tilde{v}_i &= \mathbf{c}_i x_i,\ \quad \mathbf{c}_i =\begin{bmatrix} 0 & 1\end{bmatrix}.
\end{split}
\label{eq:c matrix}
\end{align}

After substituting (\ref{eq:perturbations}) into~(\ref{eq:HV system},\ref{eq:HVcon}) and~(\ref{eq:CAV system w/ lag},\ref{eq:CCC general}), we linearize the nonlinear functions $V(.)$ and $W(.)$ by assuming that ${0 < v_i < v_{\rm max}}$, and derive the corresponding linearized models: 
\begin{align}
\begin{split}
    \dot{x}_0(t) &= \mathbf{a}_0 x_0(t) + \mathbf{b}_{0,1} \Tilde{v}_1(t) +  \sum_{k \in \Phi} \mathbf{b}_{0,k} \Tilde{v}_k(t), 
    \\
    \dot{x}_i(t) &= \mathbf{a}_i x_i(t) + \mathbf{a}_{i \tau} x_i(t-\tau) 
    \\
    &+  \mathbf{b}_i \Tilde{v}_{i+1}(t) + \mathbf{b}_{i \tau} \Tilde{v}_{i+1}(t-\tau),\ \forall i \in \{1,...,n\},
\end{split}
\label{eq:linearized dynamics}
\end{align}
with coefficient matrices ${\mathbf{a}_0,\mathbf{b}_{0,1},\mathbf{b}_{0,k}, \mathbf{a}_i, \mathbf{a}_{i \tau}, \mathbf{b}_i, \mathbf{b}_{i \tau}}$ shown in (\ref{eq:coefficient matrices}) in Appendix~\ref{app:A}.
By applying the Laplace transform with zero initial condition, (\ref{eq:c matrix},\ref{eq:linearized dynamics}) lead to:
\begin{align}
\begin{split}
    \Tilde{V}_0(s) &= T_{0,1}(s) \Tilde{V}_{1}(s) + \sum_{k \in \Phi} T_{0,k}(s) \Tilde{V}_{k}(s), 
    \\
    \Tilde{V}_1(s) &= \prod_{i=1}^{n} T_{i,i+1}(s) \Tilde{V}_{n+1}(s),\\
    \Tilde{V}_k(s) &= \prod_{i=k}^{n} T_{i,i+1}(s) \Tilde{V}_{n+1}(s).
\end{split}
\label{eq:Laplace transform}
\end{align}
Here ${\Tilde{V}_0(s)}$, ${\Tilde{V}_{1}(s)}$, ${\Tilde{V}_{k}(s)}$ and ${\Tilde{V}_{n+1}(s)}$ denote the Laplace transforms of the speed perturbation of the CAV $\Tilde{v}_0$, the preceding vehicles $\Tilde{v}_{1}$, $\Tilde{v}_{k}$ and the head vehicle $\Tilde{v}_{n+1}$, respectively, while the link transfer functions are defined as:
\begin{align}
\begin{split}
    T_{0,1}(s) &= \mathbf{c}_0 (s \mathbf{I}-\mathbf{a}_0)^{-1} \mathbf{b}_{0,1}, 
    \\
    T_{0,k}(s) &= \mathbf{c}_0 (s \mathbf{I}-\mathbf{a}_0)^{-1} \mathbf{b}_{0,k}, 
    \\
    T_{i,i+1}(s) &= \mathbf{c}_i (s \mathbf{I}-\mathbf{a}_{i}-\mathbf{a}_{i\tau} e^{-s\tau})^{-1} (\mathbf{b}_{i}+\mathbf{b}_{i\tau}e^{-s\tau}),
\end{split}
\label{eq:link transfer function general}
\end{align}
for ${i \in \{1,...,n\}}$ and ${k \in \Phi}$; see~\cite{zhang2016motif}. 
These link transfer functions can be obtained through the substitution of the coefficient matrices in (\ref{eq:coefficient matrices}), and their detailed expression can be found in (\ref{eq:link transfer function detail}) in Appendix~\ref{app:A}.

The overall response of the vehicle chain considering speed perturbations from the head vehicle ${\Tilde{V}_{n+1}(s)}$ to the tail vehicle ${\Tilde{V}_0(s)}$ can be described by:
\begin{equation}
    \Tilde{V}_0(s) =  G_{0,n+1}(s) \Tilde{V}_{n+1}(s),
\label{eq:h2t transfer func definition}
\end{equation}
where the head-to-tail transfer function can be expressed as:
\begin{align}
    G_{0,n+1}(s) \! = \! T_{0,1}(s) \! \prod_{i=1}^{n}T_{i,i+1}(s) \! + \!  \sum_{k \in \Phi} \!  \big( T_{0,k}(s) \! \prod_{i=k}^{n} \! T_{i,i+1}(s) \big).
\label{eq:h2t transfer func general}
\end{align}
Considering the case shown in Fig.~\ref{fig:3 veh setup}(a), i.e., CCC~(\ref{eq:CCC 3veh}) with ${\Phi=\{n+1\}}$, (\ref{eq:h2t transfer func general}) simplifies to:
\begin{align}
    G_{0,n+1}(s) = T_{0,1}(s)\prod_{i=1}^{n}T_{i,i+1}(s) + T_{0,n+1}(s).
\label{eq:h2t transfer func 1,n+1}
\end{align}

Using the head-to-tail transfer function, we can analyze the plant stability and head-to-tail string stability of the vehicle chain~\cite{guo2023connected}.
Plant stability is associated with the characteristic equation:
\begin{equation}
    \mathrm{D}(G_{0,n+1}(s)) = 0,
\label{eq:chara eqn general}
\end{equation}
where $\mathrm{D}(.)$ denotes the denominator. 
Plant stability is achieved when all characteristic roots have a negative real part, i.e., ${\mathrm{Re}(s_m) < 0,\ \forall m \in \mathbb{N}}$.
If a real root ${s = 0}$ is located on the imaginary axis, the system reaches the plant stability boundary:
\begin{equation}
    \mathrm{D}(G_{0,n+1}(0)) = 0.
\label{eq:plant stability boundary, s=0}
\end{equation}
Alternatively, the stability boundary is also reached if a complex conjugate pair of roots ${s = \pm \rm{j} \Omega}$, with ${\rm{j}^2 = -1}$ and ${\Omega > 0}$, is located at the imaginary axis, yielding:
\begin{align}
\begin{split}
    \mathrm{Re}(\mathrm{D}(G_{0,n+1}(\rm{j}\Omega))) &= 0, 
    \\
    \mathrm{Im}(\mathrm{D}(G_{0,n+1}(\rm{j}\Omega))) &= 0,
\end{split}
\label{eq:plant stability boundary, s=jw}
\end{align}
where $\mathrm{Re}(.)$ and $\mathrm{Im}(.)$ denote the real and imaginary part, respectively. 

Head-to-tail string stability, on the other hand, describes whether velocity fluctuations are mitigated along the vehicle chain.
According to~(\ref{eq:h2t transfer func definition}), a velocity perturbation with frequency ${\omega > 0}$ from the head vehicle will be amplified by the ratio ${\left | G_{0,n+1}(\rm{j}\omega) \right |}$ when reaching the tail vehicle. 
Therefore, head-to-tail string stability can be achieved if and only if:
\begin{equation}
    \left | G_{0,n+1}(\rm{j}\omega) \right | < 1,\ \forall \omega > 0.
\label{eq:string stability condition general}
\end{equation}
Equation~(\ref{eq:string stability condition general}) is equivalent to ${P(\omega) > 0}$ with:
\begin{equation}
    P(\omega) := \frac{1}{\omega^2} \big( \mathrm{D}(\left | G_{0,n+1}(\rm{j}\omega) \right |^2) - \mathrm{N}(\left | G_{0,n+1}(\rm{j}\omega) \right |^2) \big),
\label{eq:P(w)}
\end{equation}
where ${\mathrm{D}(.)}$ and ${\mathrm{N}(.)}$ denote denominator and numerator.
For ${\omega \rightarrow 0}$, the string stability boundaries can be obtained by:
\begin{equation}
    \lim_{\omega \rightarrow 0}P(\omega) = 0.
\label{eq:string stability boundary, w=0}
\end{equation}
For ${\omega>0}$, a set of string stability boundaries, parameterized by the wave number ${K \in [0, 2\pi)}$, is provided in~\cite{molnar2022virtual}:
\begin{equation}
    G_{0,n+1}(\mathrm{j} \omega) = \frac{n_0(\omega)+\mathrm{j}n_1(\omega)}{d_0(\omega)+\mathrm{j}d_1(\omega)} = e^{-\mathrm{j}K},
\label{eq:string stability boundary, w>0}
\end{equation}
where $n_0(\omega)$ and $n_1(\omega)$ are the real and imaginary parts of ${\mathrm{N}(G_{0,n+1}(\mathrm{j}\omega))}$ while $d_0(\omega)$ and $d_1(\omega)$ are those of ${\mathrm{D}(G_{0,n+1}(\mathrm{j}\omega))}$. 

Overall, (\ref{eq:plant stability boundary, s=0},\ref{eq:plant stability boundary, s=jw}) define the plant stability boundaries, while (\ref{eq:string stability boundary, w=0},\ref{eq:string stability boundary, w>0}) specify the head-to-tail string stability boundaries, as a function of the CCC parameters such as $A$, $B_1$, $B_k$. 
Appendix~\ref{app:A} provides the detailed formulas for these boundaries.
Using the expressions for the stability boundaries, one can construct \textit{stability charts} to illustrate the regions of the control parameters that provide plant and head-to-tail string stability for a vehicle chain, thus aiding stable controller design.

Stability charts in the ${(A,B_1)}$ and ${(B_1,B_2)}$ planes are provided in Fig.~\ref{fig: stability boundary analytical} for the scenario where a CAV responds to the CHV two vehicles ahead, using the parameters in Table~\ref{tab:parameters}. 
Note that we consider every HV to be identical with ${A_i=A_{\rm h}}$, ${B_i=B_{\rm h}}$ in (\ref{eq:HVcon}) and ${\kappa_i=\kappa_{\rm h}}$ in (\ref{eq:HV range policy}) for ${i \in \{1,...,n\}}$ in the following examples of stability charts.
In panels (a-b), the head-to-tail string stable region (blue) is enclosed by the boundaries~(\ref{eq:string stability boundary, w=0}) (blue lines) and~(\ref{eq:string stability boundary, w>0}) (rainbow colored curves for various values of $K$). 
This region can be concisely depicted as the string stable domains in panels (c-d). 
The plant stable domains (red) are also plotted below the string stable domains, with boundaries determined by~(\ref{eq:plant stability boundary, s=0},\ref{eq:plant stability boundary, s=jw}). 
Note that the vehicle chain is string stable only when it is plant stable.

\begin{figure}[t]
    \centering
    \includegraphics{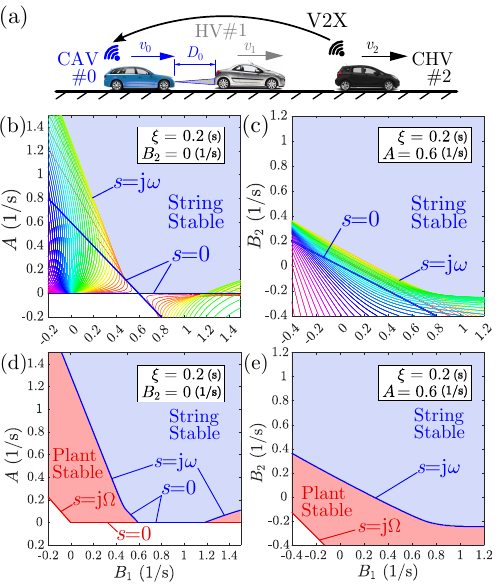}
    \vspace{-2mm}
    \caption{Stability charts of the CCC (\ref{eq:CAV system w/ lag},\ref{eq:CCC 3veh}) for ${\xi=0.2\, \rm{s}}$ and ${n=1}$ in (b) ${(A,B_1)}$ plane, (c) ${(B_1,B_2)}$ plane. (d-e) Simplified string stable domains, along with plant stable regions.}
    \label{fig: stability boundary analytical}
    \vspace{-4mm}
\end{figure}

\section{Background on Control Barrier Functions}
\label{sec:CBF}

In this section, we revisit the theory of control barrier functions, which is used to construct the safety of the CCC~(\ref{eq:CCC general}).

Consider a control-affine system: 
\begin{equation}
    \dot{x} = f(x) + g(x) u,
\label{eq:open loop sys}
\end{equation}
where ${x \in \mathbb{R}^n}$, ${u \in \mathbb{R}^m}$ denote the state and the input, respectively, and ${f: \mathbb{R}^n \to \mathbb{R}^n}$ and ${g: \mathbb{R}^n \to \mathbb{R}^{n \times m}}$ are locally Lipschitz continuous functions representing the dynamics.
The closed-loop system is established given a locally Lipschitz continuous controller ${k:\mathbb{R}^n \to \mathbb{R}^m}$, ${u=k(x)}$:
\begin{equation}
    \dot{x}=f(x)+g(x)k(x),
\label{eq:close loop sys}
\end{equation}
where $x(t)$ represents the solution starting from the initial condition ${x(t_0)}$.

A {\em safe set} $\mathcal{S}$ is defined to determine the safety of system~(\ref{eq:close loop sys}):
\begin{equation}
    \mathcal{S} = \left \{ x \in \mathbb{R}^n: h(x) \geq 0 \right \},
\label{eq:safe set}
\end{equation}
where ${h:\mathbb{R}^n \to \mathbb{R}}$ is a continuously differentiable function.
With that, system~(\ref{eq:close loop sys}) is considered safe w.r.t.~$\mathcal{S}$ if ${x(t_0) \in \mathcal{S} \implies x(t) \in \mathcal{S}}$, ${\forall t \geq 0}$, and its safety can be guaranteed by satisfying Nagumo's theorem \cite{nagumo1942lage} given below.

\begin{theorem}[\hspace{2sp}\cite{nagumo1942lage}]\label{theo:nagumo}
\textit{
Let $h$ satisfy ${\nabla h(x) \neq 0}$ for all ${x \in \mathbb{R}^n}$ such that ${h(x)=0}$. System (\ref{eq:close loop sys}) is safe w.r.t. $\mathcal{S}$ if and only if:
\begin{equation}
    \dot{h} \big( x, k(x) \big) \geq 0, \quad
    \forall x \in \mathbb{R}^n\ {\rm s.t.}\ h(x) = 0,
\label{eq:nagumo}
\end{equation}
where:
\begin{equation}
    \dot{h} \big( x, k(x) \big) = \underbrace{\nabla h(x)f(x)}_{L_fh(x)} + \underbrace{\nabla h(x)g(x)}_{L_gh(x)}k(x),
\label{eq:hdot}
\end{equation}
and ${L_fh(x)}$ and ${L_gh(x)}$ denote the Lie derivatives.
}
\end{theorem}

Condition (\ref{eq:nagumo}) is established at the boundary of $\mathcal{S}$, i.e., ${h(x)=0}$, which prevents the system from leaving the safe set. However, it does not specify the system behavior within $\mathcal{S}$, i.e., ${h(x)>0}$. To address this, {\em control barrier functions} (CBFs)~\cite{AmesXuGriTab2017} have been proposed to enable safety-critical controller synthesis for~(\ref{eq:open loop sys}).

\begin{definition}[\hspace{2sp}\cite{AmesXuGriTab2017}]
Function $h$ is a {\em control barrier function} for (\ref{eq:open loop sys}) on $\mathcal{S}$ if there exists ${\alpha \in \mathcal{K}^{\rm e}}$ such that for all ${x \in \mathcal{S}}$:
\begin{equation}
    \sup_{u \in \mathbb{R}^m} \dot{h}(x,u) > - \alpha \big( h(x) \big).
\label{eq:CBF_condition}
\end{equation}
\end{definition}

\begin{remark}
\label{re:class-k functions}
A continuous function ${\alpha: \mathbb{R} \to \mathbb{R}}$ is of extended class-$\mathcal{K}$ (${\alpha \in \mathcal{K}^{\rm e}}$) if it is strictly increasing and ${\alpha(0)=0}$.
For example, the extended class-$\mathcal{K}$ function $\alpha$ is often chosen to be a linear function, ${\alpha(r)=\gamma r}$, with some ${\gamma>0}$.
\end{remark}

\begin{theorem}[\hspace{2sp}\cite{AmesXuGriTab2017}] \label{theo:CBF}
\textit{
If $h$ is a CBF for~(\ref{eq:open loop sys}) on $\mathcal{S}$, then any locally Lipschitz continuous controller $k$ that satisfies: 
\begin{equation}
    \dot{h} \big( x, k(x) \big) \geq - \alpha \big( h(x) \big)
\label{eq:safety_condition}
\end{equation}
for all ${x \in \mathcal{S}}$ renders~(\ref{eq:close loop sys}) safe w.r.t.~$\mathcal{S}$.
}
\end{theorem}

CBFs are often employed in safety filters that minimally modify a desired but not necessarily safe controller $ {\kd:\mathbb{R}^n \to \mathbb{R}^m} $ to a safe one, by solving the following quadratic programming problem with constraint (\ref{eq:safety_condition}):
\begin{align}
\begin{split}
    k(x) = \underset{u \in \mathbb{R}^m}{\operatorname{argmin}} & \quad \| u - \kd(x) \|^2 
    \\
    \text{s.t.} & \quad \dot{h}(x,u) \geq - \alpha \big( h(x) \big).
\end{split}
\label{eq:QP}
\end{align}
The closed form solution of (\ref{eq:QP}) was given in \cite{alan2023control} as:
\begin{align}
    & k(x)= \begin{cases}
    k_{\rm d}(x) + \max \left \{0,\eta(x)\right \}\frac{L_gh(x)^\top }{\left \| L_gh(x) \right \|^2}, &\mathrm{if}\ L_gh(x) \neq 0, 
    \\ 
    \kd(x), & \mathrm{if}\ L_gh(x)=0,
    \end{cases} \nonumber
    \\
    & \eta(x) = -L_fh(x) - L_gh(x)k_{\rm d}(x) - \alpha h(x).
\label{eq:QP close form}
\end{align}
For scalar input $u$ like in (\ref{eq:CAV system w/ lag}), the safety filter simplifies to: 
\begin{align}
\begin{split}
    & k(x)= \begin{cases}
    \min \left \{ \kd(x), \ks(x) \right \}, &\mathrm{if}\ L_gh(x) \leq 0, \\ 
    \kd(x), & \mathrm{if}\ L_gh(x)=0, \\
    \max \left \{ \kd(x), \ks(x) \right \}, &\mathrm{if}\ L_gh(x) \geq 0,
    \end{cases}. \\
\end{split}
\label{eq:QP single input}
\end{align}
with:
\begin{equation}
    \ks(x)= - \frac{L_fh(x) + \alpha \big( h(x) \big)}{L_gh(x)}.
\label{eq:ks general}    
\end{equation}

Note that when ${L_g h(x) \equiv 0}$, the safety of the system (\ref{eq:close loop sys}) is not directly affected by the controller $k(x)$, according to~(\ref{eq:hdot}).
Therefore, $h$ is neither a valid CBF nor applicable for synthesizing safety-critical controllers. 
In such a case, an {\em extended CBF} can be constructed~\cite{Nguyen2016,xiao2019cbf,cohen2024ROMCBF}:
\begin{equation}
    h_{\rm e}(x) = L_f h(x) + \alpha \big( h(x) \big),
\label{eq:CBFextension}
\end{equation}
which is associated with the {\em extended safe set} $\mathcal{S}_{\rm e}$:
\begin{equation}
    \mathcal{S}_{\rm e} = \{x \in \mathbb{R}^n: h_{\rm e}(x) \geq 0 \}.
\label{eq:safeset_extended}
\end{equation}
Condition~(\ref{eq:safety_condition}) is satisfied when the system remains within $\mathcal{S}_{\rm e}$.
Therefore, safety w.r.t.~the intersection ${\mathcal{S} \cap \mathcal{S}_{\rm e}}$ of the two sets can be established as follows.

\begin{corollary}[\hspace{2sp}\cite{xiao2019cbf}] \label{cor:extendedCBF}
\textit{
If ${L_g h(x) \equiv 0}$ and $h_{\rm e}$ in~(\ref{eq:CBFextension}) is a CBF for~(\ref{eq:open loop sys}) on $\mathcal{S}_{\rm e}$ with ${\alpha_{\rm e} \in \mathcal{K}^{\rm e}}$, then any locally Lipschitz continuous controller $k$ that satisfies: 
\begin{equation}
    \dot{h}_{\rm e} \big( x, k(x) \big) \geq - \alpha_{\rm e} \big(  h_{\rm e}(x) \big)
\label{eq:safety_condition_extended}
\end{equation}
for all ${x \in \mathcal{S}_{\rm e}}$ renders~(\ref{eq:close loop sys}) safe w.r.t.~${\mathcal{S} \cap \mathcal{S}_{\rm e}}$.
}
\end{corollary}
\noindent Note that similar to Theorem~\ref{theo:CBF}, $\alpha_{\rm e}$ can be chosen, for example, to be linear as ${\alpha_{\rm e}(r) = \gamma_{\rm e}r}$ with ${\gamma_{\rm e} > 0}$.

Accordingly, Nagumo's theorem~(\ref{eq:nagumo}) is conducted at the boundary of ${\mathcal{S} \cap \mathcal{S}_{\rm e}}$, specifically when ${h(x) = 0}$ and ${h_{\rm e}(x) \geq 0}$, or when ${h_{\rm e}(x) = 0}$ and ${h(x) \geq 0}$. Note that the former case implies ${\dot{h} \big( x,k(x) \big) \geq 0}$.
It has already satisfied condition~(\ref{eq:nagumo}), and thereby requires no further analysis. 
The safety condition for the latter case is given as follows.
\begin{corollary}[\hspace{2sp}\cite{Tamas2023CDC}] \label{cor:extendedNagumo}
\textit{
Let ${L_g h(x) \equiv 0}$ and $h_{\rm e}$ in~(\ref{eq:CBFextension}) satisfy ${\nabla h_{\rm e}(x) \neq 0}$ for all ${x \in \mathbb{R}^n}$ such that ${h_{\rm e}(x) = 0}$.
System~(\ref{eq:close loop sys}) is safe w.r.t.~${\mathcal{S} \!\cap \!\mathcal{S}_{\rm e}}$ if:
\begin{align}
\begin{split}
    \dot{h}_{\rm e} \big( x, k(x) \big) \geq 0, \quad
    \forall x \in \mathbb{R}^n\ {\rm s.t.}\ h_{\rm e}(x) = 0\ \mathrm{and}\ h(x) \geq 0.
\end{split}
\label{eq:Nagumo_extended}
\end{align}
}
\end{corollary}
For scalar input $u$, as in system~(\ref{eq:CAV system w/ lag}), 
(\ref{eq:Nagumo_extended}) is equivalent to:
\begin{align}
\begin{split}
    \ks(x)-\kd(x) \geq 0, \quad
    \forall x \in \mathbb{R}^n\ {\rm s.t.}\ h_{\rm e}(x) = 0\ \mathrm{and}\ h(x) \geq 0.
\end{split}
\label{eq:Nagumo Lgh<0}
\end{align}
where:
\begin{equation}
    \ks(x) = -\frac{L_f h_{\rm e}(x) + \alpha_{\rm e} \big( h_{\rm e}(x) \big)}{L_g h_{\rm e}(x)},
\label{eq:ks_extended}
\end{equation}
cf.~(\ref{eq:ks general}), and formula~(\ref{eq:QP single input}) can be used as safety filter.

\section{Safe Connected Cruise Control}
\label{sec:safccc}

In this section, we first investigate the safety of the nominal CCC~(\ref{eq:CCC general}) via Corollary~\ref{cor:extendedNagumo}, then modify the nominal CCC to propose a safety-critical CCC law via Corollary~\ref{cor:extendedCBF}.

\subsection{Safe Nominal CCC Design}

Now we analyze the safety of the CAV dynamics~(\ref{eq:CAV system w/ lag}) executing the nominal CCC~(\ref{eq:CCC general}) based on condition~(\ref{eq:Nagumo Lgh<0}), and propose a guideline for safe nominal CCC design.

First, we append the expression of the front vehicle's acceleration, i.e., ${\dot{v}_{1}(t)  = a_{1}(t)}$, to the CAV dynamics~(\ref{eq:CAV system w/ lag}), and write the system in the form of~(\ref{eq:open loop sys}) with:
\begin{align}
\begin{split}
x=\begin{bmatrix}
    D_0 \\
    v_0 \\
    a_0 \\
    v_1
\end{bmatrix},\quad
f(x)=\begin{bmatrix}
    v_{1}-v_{0} \\
    a_{0} \\
    -\frac{a_{0}}{\xi} \\
    a_{1}
    \end{bmatrix},\quad  
g(x)=\begin{bmatrix}
    0 \\ 
    0 \\ 
    \frac{1}{\xi} \\
    0
    \end{bmatrix}.
\end{split}
\label{eq:control affine system}
\end{align}
Then, we select function $h$ to characterize the CAV's safety.
Several safety criteria were introduced in \cite{Tamas2023CDC}. In this analysis, we adopt the strictest criterion, which mandates maintaining the {\em time headway} $D_0/v_0$ of the CAV above a specified safe threshold defined by ${T_{\rm h} = 1/\kappa_{\rm{sf}}}$:
\begin{equation}
    h(x) = \kappa_{\rm{sf}} (D_0 - \Dsf) - v_0,
\label{eq:TH}
\end{equation}
where $\Dsf$ denotes the safe standstill distance.
This leads to:
\begin{align}
\begin{split}
    \nabla h(x) &= \begin{bmatrix} \kappa_{\rm{sf}} & -1 & 0 & 0\end{bmatrix} \neq 0,\\
    L_f h(x) &= \ksf (v_{1} - v_{0}) - a_{0},\quad L_g h(x) \equiv  0.
\end{split}
\label{eq:CBF derivative detail}
\end{align}

Since ${L_g h(x) \equiv  0}$, the extended CBF in~(\ref{eq:CBFextension}) can be constructed as:
\begin{equation}
    h_{\rm e}(x) = \ksf (v_{1} - v_{0}) - a_{0} + \gamma \big( \ksf (D_0 - \Dsf) - v_0 \big),
\label{eq:extend TH}
\end{equation}
where ${\gamma > 0}$ denotes the coefficient of the linear extended class-$\mathcal{K}$ function ${\alpha(r)=\gamma r}$.
Then, we have: 
\begin{align}
    \nabla h_{\rm e}(x) &= \begin{bmatrix} \gamma \ksf & -\ksf-\gamma & -1 & \ksf\end{bmatrix} \neq 0, \nonumber
    \\
    L_f h_{\rm e}(x) &= \ksf (a_{1} - a_{0}) + a_{0}/\xi + \gamma \big( \ksf (v_{1} - v_{0}) - a_{0} \big), \nonumber
    \\
    L_g h_{\rm e}(x) &= -1/\xi < 0.
\label{eq:eCBF derivative detail}
\end{align}

Hence, Corollary~\ref{cor:extendedCBF} and \ref{cor:extendedNagumo} hold, i.e., (\ref{eq:Nagumo Lgh<0}) can be applied for safety analysis, whereas the safety filter (\ref{eq:QP single input}) becomes:
\begin{equation}
    k(x) = \min\{\kd(x),\ks(x)\},
\label{eq:SafCon}
\end{equation}
while (\ref{eq:ks_extended}) reads:
\begin{align}
\begin{split}
    \ks(x) = (1 - \xi \ksf) a_{0} + \xi \ksf a_{1} + \xi \gamma \big( \ksf (v_{1} - v_{0}) - a_{0} \big) & \\
    + \xi \gamma_{\rm e} \Big( \ksf (v_{1} - v_{0}) - a_{0} + \gamma \big( \ksf (D_0 - \Dsf) - v_0 \big) \Big) &.
\end{split}
\label{eq:extended SafFilter}
\end{align}
where ${\gamma_{\rm e} > 0}$ denotes the coefficient of the linear extended class-$\mathcal{K}$ function ${\alpha_{\rm e}(r)=\gamma_{\rm e} r}$.

To analyze the safety of the nominal CCC (\ref{eq:CCC general}), we apply Corollary~\ref{cor:extendedNagumo}. 
Specifically, by analyzing when condition (\ref{eq:Nagumo Lgh<0}) holds, safe choices of controller gains $A$, $B_1$ and ${B_k,\ \forall k \in \Phi}$ can be derived as follows.

\begin{figure}[t]
    \centering
    \includegraphics{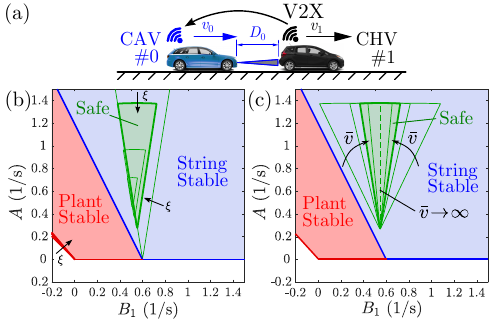}
    \vspace{-4mm}
    \caption{Safety charts of the nominal CCC (\ref{eq:CCC general}) when the CAV responds to the vehicle immediately ahead of it; see panel (a). The time headway criterion~(\ref{eq:TH}) is considered for (b) different CAV lags $\xi$ in system (\ref{eq:CAV system w/ lag}) and (c) different bounds $\bar{v}$ on the speed difference in (\ref{eq:Saf Region}) with ${\xi=0.15\, \rm{s}}$.}
    \label{fig:2veh Safety Chart}
    \vspace{-6mm}
\end{figure}

\begin{theorem}
\label{theo:Saf Region A-B1-BN}
\textit{
System (\ref{eq:CAV system w/ lag}) with ${u_0 = \kd(x)}$ given by (\ref{eq:CCC general}) and ${A \geq 0}$, ${B_1 \geq 0}$, ${B_k \geq 0,\ \forall k \in \Phi}$ 
is safe ${w.r.t.\ \mathcal{S} \! \cap \! \mathcal{S}_{\rm{e}}}$ given by (\ref{eq:safe set},\ref{eq:safeset_extended},\ref{eq:TH},\ref{eq:extend TH}) if ${v_0 \geq 0}$, ${|v_1 - v_0| \! \leq \! \bar{v}}$, ${|v_{k} - v_0| \! \leq \! \bar{v}},\ {\forall k \! \in \! \Phi}$ with some ${\bar{v}>0}$, ${a_1 \geq - a_{\rm{min}}}$ with some ${a_{\rm{min}}>0}$, ${\Dst > \Dsf}$, ${\kappa_{\rm{sf}} \geq \kappa > 0}$, ${\gamma > 0}$ and:
\begin{align}
\begin{split}
    A &\leq \frac{(1-\xi \ksf)^2}{4 \xi} - \xi \bigg(\gamma - \frac{1-\xi\ksf}{2\xi} \bigg)^2, \\
    A &\geq \frac{\big( |\ksf-\xi \ksf^2 -B_1| + \sum_{k \in \Phi} B_k \big)\bar{v} + \xi \ksf a_{\rm{min}}}{\kappa(\Dst-\Dsf)}.
\end{split}
\label{eq:Saf Region}
\end{align}
}
\end{theorem}

\begin{remark}
The upper bound of $A$ in Theorem~\ref{theo:Saf Region A-B1-BN} depends on $\gamma$, which can be chosen as ${\gamma = (1-\xi\ksf)/(2\xi)}$ to achieve the highest upper bound, i.e., largest region of safe control gains. In this case, ${1/\xi > \ksf}$ should be satisfied to guarantee ${\gamma > 0}$. Unless specified otherwise, we consider ${\gamma = (1-\xi\ksf)/(2\xi)}$ and ${1/\xi > \ksf}$ in the following calculations.
\end{remark}

Safe choices of control gains do not exist if the lag $\xi$ exceeds a critical value $\xi_{\rm cr}$. 
This value can be derived using (\ref{eq:Saf Region}).
\begin{corollary} \label{cor:critical lag}
\textit{
Safe control gains exist only if:
\begin{equation}
    \label{eq:critical lag}
    \xi \leq \xi_{\rm cr} = \frac{1}{\ksf + 2 \sqrt{\frac{\ksf a_{\rm{min}}}{\kappa (\Dst-\Dsf)}}}.
\end{equation}
}
\end{corollary}

\noindent Please refer to Appendix~\ref{app:B} and Appendix~\ref{app:C} for the proofs of Theorem~\ref{theo:Saf Region A-B1-BN} and Corollary~\ref{cor:critical lag}.

\begin{remark}
Based on (\ref{eq:Saf Region}), for unbounded velocity difference ${\bar{v} \rightarrow \infty}$, the safe region becomes the single line segment:
\begin{align}
\begin{split}
    &\frac{\xi \ksf a_{\rm{min}}}{\kappa(\Dst-\Dsf)} \leq A \leq \frac{(1-\xi \kappa)^2}{4 \xi}, \\
    &B_1 = \ksf-\xi \ksf^2,\ B_k = 0,\ \forall k \in \Phi.
\end{split}
\label{eq:Saf Region v->inf}
\end{align}
\end{remark}

\begin{figure}[t]
    \centering
    \includegraphics{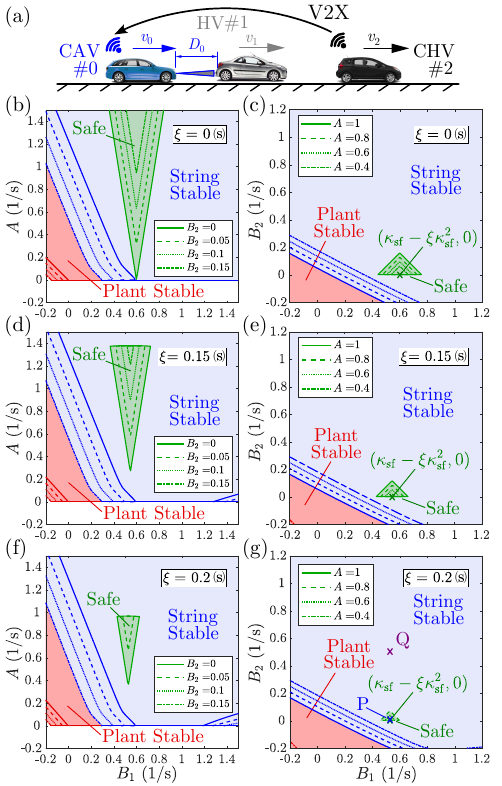}
    \vspace{-2mm}
    \caption{Safety charts of the nominal CCC (\ref{eq:CCC 3veh}) when the CAV responds to the CHV two vehicles ahead ${(n=1)}$, see panel (a).
    The time headway criterion~(\ref{eq:TH}) is considered for the lag (b-c) ${\xi = 0\, \rm{s}}$; (d-e) ${\xi = 0.15\, \rm{s}}$; (f-g) ${\xi = 0.2\, \rm{s}}$ in CAV system~(\ref{eq:CAV system w/ lag}).}
    \label{fig:3veh Safety Chart}
    \vspace{-4mm}
\end{figure}

Fig.~\ref{fig:2veh Safety Chart}(b-c) depicts the safety condition (\ref{eq:Saf Region}) for the situation in Fig.~\ref{fig:2veh Safety Chart}(a) when the CAV only responds to the preceding CHV using CCC (\ref{eq:CCC general}) with ${\Phi = \emptyset}$.
They are represented under varying lags ($\xi=0\, \rm{s},0.15\, \rm{s},0.2\, \rm{s},0.25\, \rm{s}$) and speed difference limits ($\bar{v}=5\, \rm{m/s},10\, \rm{m/s},15\, \rm{m/s},25\, \rm{m/s}$), respectively, using the parameters from Table~\ref{tab:parameters}. 
These plots are called \textit{safety charts}~\cite{He2018}\cite{Tamas2023CDC}, where the safe domain (green) indicates safe choices of CCC parameters. 
For this case, the safe region is triangular. 
In panel (b), as the lag increases, the upper horizontal boundary is moving down and the lower V-shaped boundary is moving towards smaller $B_1$ and larger $A$, which makes the safe region shrink until it becomes a point at the critical lag ${\xi_{\rm{cr}}=0.3\, \rm{s}}$ obtained by (\ref{eq:critical lag}). 
In panel (c), as the speed difference limit $\bar{v}$ increases, the V-shaped safe region gradually narrows to the green dashed line given by (\ref{eq:Saf Region v->inf}).
This suggests that larger speed difference between the CAV and the CHV makes it more difficult to find safe CCC parameters.

The plant and head-to-tail string stability domains are also visualized using red and blue shading, respectively, in Fig.~\ref{fig:2veh Safety Chart}(b-c).
Their boundaries are given by (\ref{eq:plant stability boundary s=0 final (A,B1)},\ref{eq:plant stability boundary s=jw final (A,B1)}) for plant stability and (\ref{eq:string stability boundary w=0 (A,B1)},\ref{eq:string stability boundary final (A,B1)}) for head-to-tail string stability with ${B_{n+1}=0}$, see Appendix~\ref{app:A}. 
Notably, except for extremely small $\bar{v}$, the safe domain consistently lies within the plant and string stable region, indicating that safe CCC parameters also achieve plant and string stability for this case.

To show how V2X connectivity affects the safety of the CAVs, we examine a scenario presented in Fig.~\ref{fig:3veh Safety Chart}(a), where the CAV responds not only to the preceding HV but also to the CHV two vehicles ahead using CCC (\ref{eq:CCC 3veh}) with $n=1$.
The safety condition (\ref{eq:Saf Region}) is visualized in the ${(B_1, A)}$ plane for various $B_2$ values, and in the ${(B_1, B_2)}$ plane for various $A$ values, using parameters from Table~\ref{tab:parameters} and different lags (${\xi=0\, \rm{s}}$, ${0.15\, \rm{s}}$, and ${0.2\, \rm{s}}$), see Fig.~\ref{fig:3veh Safety Chart}(b-g). 
Similar to Fig.~\ref{fig:2veh Safety Chart}(b), the safe region (green) in Fig.~\ref{fig:3veh Safety Chart}(b,d,f) is still triangular and it shrinks as the lag increases, ultimately vanishing at the same critical lag ${\xi=0.3\, \rm{s}}$. 
Indeed, the safe domains for ${B_2=0}$ match those presented in Fig.~\ref{fig:2veh Safety Chart}(b).
Furthermore, the safe domain moves towards higher $A$ values as $B_2$ increases, signifying that enhanced response to the head CHV’s speed makes a provably safe nominal CCC design more challenging.  
In Fig.~\ref{fig:3veh Safety Chart}(c,e,g), the safe region in the ${(B_1, B_2)}$ plane shrinks to the "center" at ${(\ksf - \xi\ksf^2, 0)}$ as $A$ decreases.
The plant stable and head-to-tail string stable domains are also presented with red and blue shading, respectively. 
Similar to Fig.~\ref{fig:2veh Safety Chart}, the safe domains in Fig.~\ref{fig:3veh Safety Chart} consistently reside within the plant and string stable region.

We remark that, in addition to CCC~(\ref{eq:CCC general}), a more general form of CCC could also be used, where the CAV not only responds to distance and speeds, but also to the HV's and CHV's accelerations~\cite{Ge2014}:
\begin{align}
\begin{split}
    \kd(x) &= A \big( V(D_0) - v_0 \big) + B_1 \big( W(v_{1}) - v_0 \big)\\ 
    &+ \sum_{k \in \Phi} B_k \big( W(v_{k}) - v_0 \big) + C_1 a_{1} + \sum_{k \in \Phi} C_k a_k.
\end{split}
\label{eq:CCC general w/ acc}
\end{align}
where $C_1$ and $C_k$ are acceleration feedback gains.

Similar to Theorem~\ref{theo:Saf Region A-B1-BN}, we derive the safety condition of CCC~(\ref{eq:CCC general w/ acc}) by analyzing when condition (\ref{eq:Nagumo Lgh<0}) holds.

\begin{theorem}
\label{theo:Saf Region general CCC}
\textit{
System (\ref{eq:CAV system w/ lag}) with ${u_0 = \kd(x)}$ given by (\ref{eq:CCC general w/ acc}) and ${A \geq 0}$, ${B_1 \geq 0}$, ${B_k \geq 0,\ \forall k \in \Phi}$, is safe ${w.r.t.\ \mathcal{S} \cap \mathcal{S}_{\rm{e}}}$ given by~(\ref{eq:safe set},\ref{eq:safeset_extended},\ref{eq:TH},\ref{eq:extend TH}) if ${v_0 \geq 0}$, ${|v_{1} - v_0| \leq \bar{v}}$, ${|v_{k} - v_0| \leq \bar{v},\ \forall k \in \Phi}$ with some ${\bar{v}>0}$, ${a_{1}, a_{k} \in [-\bar{a}, \bar{a}],\ \forall k \in \Phi}$ with some ${\bar{a}>0}$, ${\Dst > \Dsf}$, ${\kappa_{\rm{sf}} \geq \kappa}$, ${\gamma > 0}$, and:
\begin{align}
\begin{split}
    \frac{N_1 \bar{v} + N_2 \bar{a}}{\kappa(\Dst-\Dsf)} \! \leq \! A \! \leq \! \frac{(1-\xi \ksf)^2}{4 \xi} \! - \! \xi \bigg(\gamma \!-\! \frac{1-\xi\ksf}{2\xi} \bigg)^2,
\end{split}
\label{eq:Saf Region w/ acc}
\end{align}
where:
\begin{align}
\begin{split}
    & N_1 = |\ksf-\xi \ksf^2 -B_1| + \sum_{k \in \Phi} B_k, \\
    & N_2 = | \xi \ksf - C_1 | + \sum_{k \in \Omega} |C_k|.
\end{split}
\label{eq:Saf Region w/ acc detail}
\end{align}
}
\end{theorem}

\noindent Please refer to the Appendix~\ref{app:D} for the proof.

\subsection{Safety-critical CCC}

The safety charts in Fig.~\ref{fig:2veh Safety Chart} and Fig.~\ref{fig:3veh Safety Chart} help determine safe controller parameters for the nominal CCC (\ref{eq:CCC general}).
However, the safe region is relatively small compared to the string stable domain.
Safe regions in Fig.~\ref{fig:3veh Safety Chart}(c,e,g) only cover low $B_2$ values, which prevent safe CCC from harnessing the benefits of connectivity, e.g., energy efficiency or high degree of string stability \cite{chen2024CCCsafety}.
Furthermore, the safe domain even shrinks considerably as the lag increases, making it even impossible to achieve a safe nominal CCC design when ${\xi > \xi_{\rm cr}}$.
Therefore, tuning the existing CCC to always maintain safety is overly conservative and may lead to sub-optimal performance.
To mitigate the trade-off between safety and performance, we propose the \textit{safety-critical CCC} defined by (\ref{eq:CCC general},\ref{eq:SafCon},\ref{eq:extended SafFilter}). 
This approach enables us to optimize the gains of the nominal CCC (\ref{eq:CCC general}) to achieve good performance, while the safety filter (\ref{eq:SafCon}) activates when necessary to guarantee safety.

To evaluate the proposed controller, we first simulate a scenario where the CAV responds to the CHV two vehicles ahead (${n=1}$), as illustrated in Fig.~\ref{fig:3veh Safety Chart}(a). This simulation incorporates a lag of ${\xi=0.2\, \rm{s}}$ and tests various controllers:
\begin{itemize}
    \item nominal CCC~(\ref{eq:CCC 3veh}) with safe gains (point P in Fig.~\ref{fig:3veh Safety Chart}(g)),
    \item nominal CCC~(\ref{eq:CCC 3veh}) with unsafe gains (point Q in Fig.~\ref{fig:3veh Safety Chart}(g)),
    \item safety-critical CCC (\ref{eq:CCC 3veh},\ref{eq:SafCon},\ref{eq:extended SafFilter}) with unsafe gains (point Q).
\end{itemize}
The results are shown in Fig.~\ref{fig:3veh Simulation safe/unsafe gains}(a-d) by teal, blue and orange dashed curves, respectively, using parameters from Table ~\ref{tab:parameters}.

\begin{figure}[ht]
    \centering
    \includegraphics{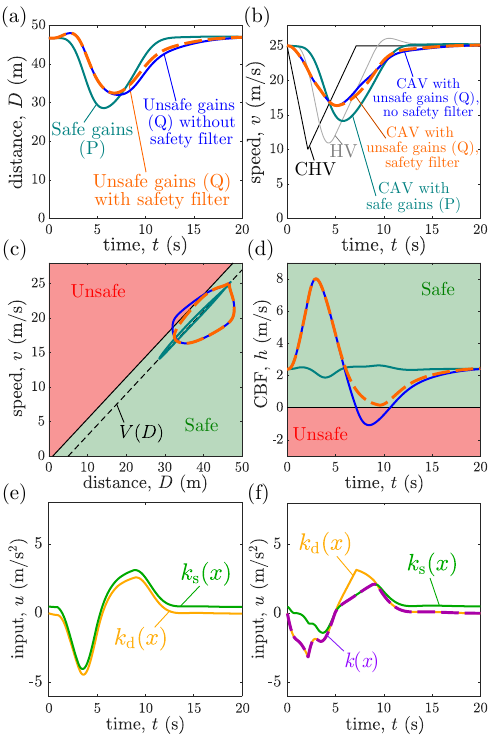}
    \vspace{-3mm}
    \caption{(a-d) Simulations of system (\ref{eq:HV system},\ref{eq:CAV system w/ lag}) for ${\xi = 0.2\, \mathrm{s}}$ using the CCC~(\ref{eq:CCC 3veh}) with ${n=1}$ and with safe gains (teal) and unsafe gains (blue), and using the safety-critical CCC~(\ref{eq:CCC 3veh},\ref{eq:SafCon},\ref{eq:extended SafFilter}) with a safety filter (orange).
    Nominal and safe control inputs using (e) safe gains, and (f) unsafe gains with safety filter.}
    \label{fig:3veh Simulation safe/unsafe gains}
    \vspace{-6mm}
\end{figure}

In Fig.~\ref{fig:3veh Simulation safe/unsafe gains}, we consider an case where the CHV applies the brake with maximum deceleration, $a_{\rm min}$, and subsequently resumes its initial speed using maximum acceleration, $a_{\rm max}$, causing a speed perturbation denoted as $v_{\rm pert}$. 
The behaviors expected from Fig.~\ref{fig:3veh Safety Chart}(f-g) are shown in panel (b).
In all cases the CAV is plant and head-to-tail string stable, as its speed fluctuation is smaller than that of the CHV.
In panels (c-d), the trajectory for safe gains P (teal) always stays inside the safe set, while the path for unsafe gains Q (blue) crosses the safety boundary and enters the unsafe region.
Additionally, as stated by Corollary~\ref{cor:extendedCBF}, the safety-critical CCC~(\ref{eq:CCC 3veh},\ref{eq:SafCon},\ref{eq:extended SafFilter}) successfully guarantees safety by slightly altering the unsafe nominal CCC only when the system is about to become unsafe (orange).

This minimal safety filter intervention is further demonstrated by Fig.~\ref{fig:3veh Simulation safe/unsafe gains}(e-f), which presents the nominal CCC input~(\ref{eq:CCC 3veh}) and safe controller input (\ref{eq:extended SafFilter}), denoted as $k_{\rm d}(x)$ and $k_{\rm s}(x)$, respectively.
For safe gains P in panel (e), $\ks(x)$ remains consistently higher than $\kd(x)$, indicating that the nominal CCC $\kd(x)$ can ensure safety and the safety filter does not need to be activated.
For unsafe gains Q in panel (f), the safety-critical controller $k(x)$ remains the same as $k_{\rm d}(x)$ when safe, while safety filter~(\ref{eq:SafCon}) intervenes when $k_{\rm d}(x)$ exceeds $k_{\rm s}(x)$ and alters the control input to $k_{\rm s}(x)$ to ensure safety.

The safety filter is activated during acceleration rather than deceleration, see the orange dashed curve in Fig.~\ref{fig:3veh Simulation safe/unsafe gains}(b) and the purple dashed curve in Fig.~\ref{fig:3veh Simulation safe/unsafe gains}(f). 
This occurs since the CAV matches its velocity to the CHV.
As the CHV decelerates, the CAV responds to it and brakes sooner than the HV, thereby increasing the distance between them and improving safety, see orange dashed curves in Fig.~\ref{fig:3veh Simulation safe/unsafe gains}(a) and~(d).
However, when the CHV accelerates after braking, the CAV also begins to accelerate to match its velocity even when the preceding HV still drives at a lower speed, resulting in rapidly decreasing distance, see Fig.~\ref{fig:3veh Simulation safe/unsafe gains}(a-b).
Although this could cause a safety violation for the nominal CCC (blue), the safety filter engages at this point, mitigates the acceleration, and maintains a safe distance to the HV, see Fig.~\ref{fig:3veh Simulation safe/unsafe gains}(f).
Therefore, with the proposed safety-critical controller, connectivity can be leveraged to improve safety during deceleration while avoiding safety violations during acceleration.

\begin{figure}[t]
    \centering
    \includegraphics{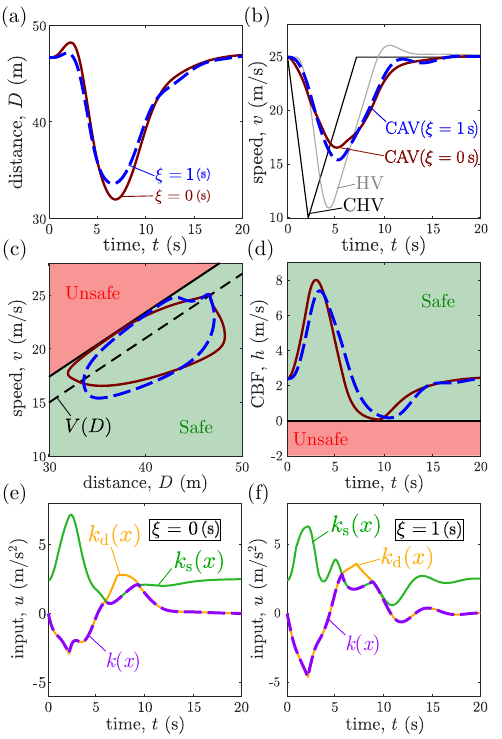}
    \vspace{-3mm}
    \caption{(a-d) Simulations of system~(\ref{eq:HV system},\ref{eq:CAV system w/ lag}) for lags ${\xi=0\, \mathrm{s}}$ (dark red) and ${\xi=1\, \mathrm{s}}$ (dashed blue) using the safety-critical CCC~(\ref{eq:CCC 3veh},\ref{eq:SafCon},\ref{eq:extended SafFilter}) with ${n=1}$ and unsafe gains. Nominal and safe control inputs with  (e) $\xi=0\, \mathrm{s}$, and with (f) $\xi=1\, \mathrm{s}$.}
    \label{fig:3veh Simulation w/ different lag}
    \vspace{-6mm}
\end{figure}

To investigate the impact of the first-order lag $\xi$ on the behavior of safety-critical CCC~(\ref{eq:CCC 3veh},\ref{eq:SafCon},\ref{eq:extended SafFilter}), we simulate same scenario depicted in Fig.~\ref{fig:3veh Simulation safe/unsafe gains}. 
These simulations incorporate lags ${\xi=0\, \rm{s}}$ and ${\xi=1\, \rm{s}}$ in~(\ref{eq:CAV system w/ lag}), see Fig.~\ref{fig:3veh Simulation w/ different lag}. 
We use the parameters in Table \ref{tab:parameters} and the unsafe gains Q with the safety filter.

In Fig.~\ref{fig:3veh Simulation w/ different lag}(a-b), when the CHV decelerates, the CAV with a higher lag (blue) responds later and requires more pronounced braking to maintain a safe distance, which results in a larger speed perturbation.
However, when the CHV accelerates subsequent to braking, the CAV with a higher lag also accelerates later, and at the same time the HV also starts to increase its speed. 
Thus, the lag in acceleration enlarges the minimum distance between the CAV and HV, as shown in Fig.~\ref{fig:3veh Simulation w/ different lag}(a), thereby improving safety during acceleration and requiring less safety filter intervention, see Fig.~\ref{fig:3veh Simulation w/ different lag}(e-f). 
Importantly, the safety filter can always intervene and provide safety guarantees regardless of lags, as the trajectories always keep within the safe region in Fig.~\ref{fig:3veh Simulation w/ different lag}(c-d). 
Fig.~\ref{fig:3veh Simulation w/ different lag}(e-f) demonstrate the safety filter intervention by showing the nominal CCC input $k_{\rm d}(x)$ in (\ref{eq:CCC general}) and safe controller input $k_{\rm s}(x)$ in (\ref{eq:extended SafFilter})  with CAV lag ${\xi=0\, \rm{s}}$ and ${\xi=1\, \rm{s}}$, respectively. 
In both panels, the safety filter engages during CAV acceleration to mitigate the acceleration and keep safe distance. 
For the CAV with larger lag ${\xi=1\, \rm{s}}$, it intervenes later and for a shorter period due to the CAV's later response to the CHV's acceleration.

To show the applicability of the proposed controller in the real world, we explore the multi-vehicle scenarios shown in Fig.~\ref{fig:real data Simulation n=2}(a) and Fig.~\ref{fig:real data Simulation n=6}(a), where the CAV is connected to a CHV two vehicles ahead (${n=1}$, ${\Phi={2}}$) and six vehicles ahead (${n=5}$, ${\Phi={6}}$). 
The speed profiles of the CHV and HVs are from real experimental data~\cite{jin2018experimental}. 
Fig.~\ref{fig:real data Simulation n=2}(b-c) and Fig.~\ref{fig:real data Simulation n=6}(b-c) present the simulation results of CAVs with lag ${\xi=0.6\rm{s}}$ using nominal CCC~(\ref{eq:CCC 3veh}) (blue) and safety-critical CCC~(\ref{eq:CCC 3veh},\ref{eq:SafCon},\ref{eq:extended SafFilter}) (orange), both with unsafe gains Q.
Fig.~\ref{fig:real data Simulation n=2}(d) and Fig.~\ref{fig:real data Simulation n=6}(d) showcase the safety filter's engagement by visualizing the nominal, safe, and safety-critical control inputs, denoted as $k_{\rm d}(x)$, $k_{\rm s}(x)$, and $k(x)$, respectively.

In Fig.~\ref{fig:real data Simulation n=2} and Fig.~\ref{fig:real data Simulation n=6}, the safety-critical CCC successfully prevents the CAV from leaving the safe set as opposed to the nominal CCC, see panel (c). 
Similar to Fig.~\ref{fig:3veh Simulation safe/unsafe gains} and Fig.~\ref{fig:3veh Simulation w/ different lag}, the safety filter engages more frequently during CAV acceleration than deceleration, see panels (b) and (d) in Fig.~\ref{fig:real data Simulation n=2} and Fig.~\ref{fig:real data Simulation n=6}. 
During deceleration, the CAV that is connected to a CHV farther ahead initiates braking earlier to synchronize with the CHV's speed, consequently increasing the distance and enhancing safety. 
This behavior results in the higher peak of $k_{\rm s}(x)$ in Fig.~\ref{fig:real data Simulation n=6}(d) compared to Fig.~\ref{fig:real data Simulation n=2}(d).
During acceleration, the CAV that is connected to a CHV farther ahead also accelerates earlier while the preceding HV is still traveling at a lower speed, as shown by the gray line in the subplot of Fig.~\ref{fig:real data Simulation n=6}(b).
This leads to a rapidly decreasing distance, necessitating more safety filter intervention to keep a safe distance, as indicated in the subplot of Fig.~\ref{fig:real data Simulation n=6}(b).

The prolonged safety filter interventions are also evident from the longer duration when ${k_{\rm d}(x) > k_{\rm s}(x)}$ in Fig.~\ref{fig:real data Simulation n=6}(d) compared to Fig.~\ref{fig:real data Simulation n=2}(d). 
Moreover, the behavior of $k_{\rm d}(x)$ in Fig.~\ref{fig:real data Simulation n=2}(d) and Fig.~\ref{fig:real data Simulation n=6}(d) shows that connecting to a farther CHV results in smoother deceleration and harsher acceleration. 
However, the safety filter significantly mitigates the acceleration, thereby reducing the control effort required for the CAV, lowering energy consumption, and enhancing comfort, as seen in $k(x)$.
To summarize, connecting to a farther CHV substantially enhances safety during deceleration, and it may force the safety filter to intervene when ensuring safety during acceleration. 
Consequently, V2X connectivity, together with a safety-critical controller, significantly improves safety and reduces the control effort for the longitudinal control of CAVs.

\begin{figure*}[!t]
    \centering
    \includegraphics{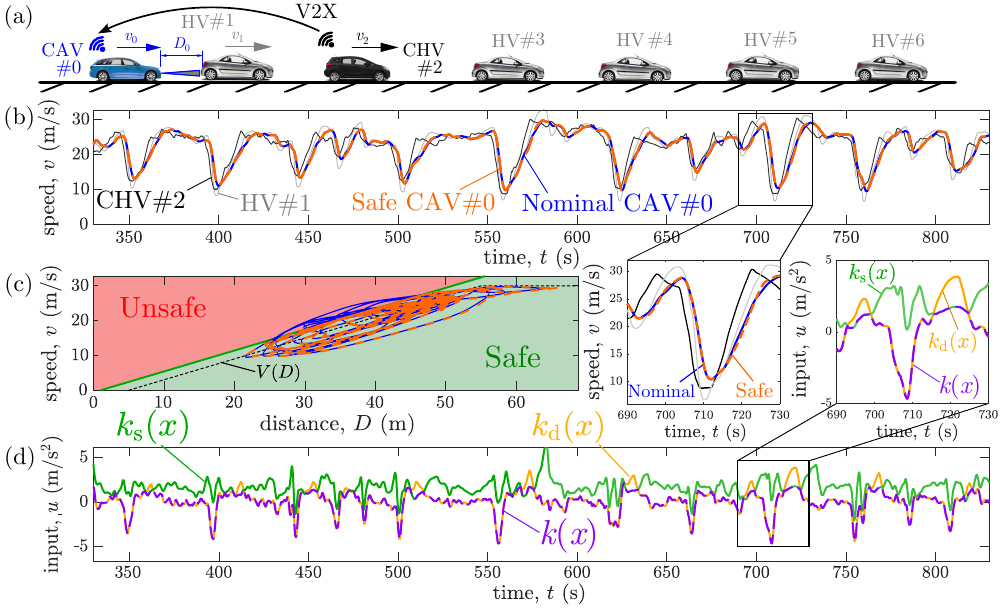}
    \vspace{-3mm}
    \caption{Simulation of a vehicle chain where a CAV (\ref{eq:CAV system w/ lag}) responds to the CHV two vehicles ahead with lag ${\xi=0.6\, \rm{s}}$ using the safety-critical CCC~(\ref{eq:CCC 3veh},\ref{eq:SafCon},\ref{eq:extended SafFilter}). The speed profiles of the CHV and the HVs are real experiment data.}
    \label{fig:real data Simulation n=2}
\end{figure*}

\begin{figure*}[!t]
    \centering
    \includegraphics{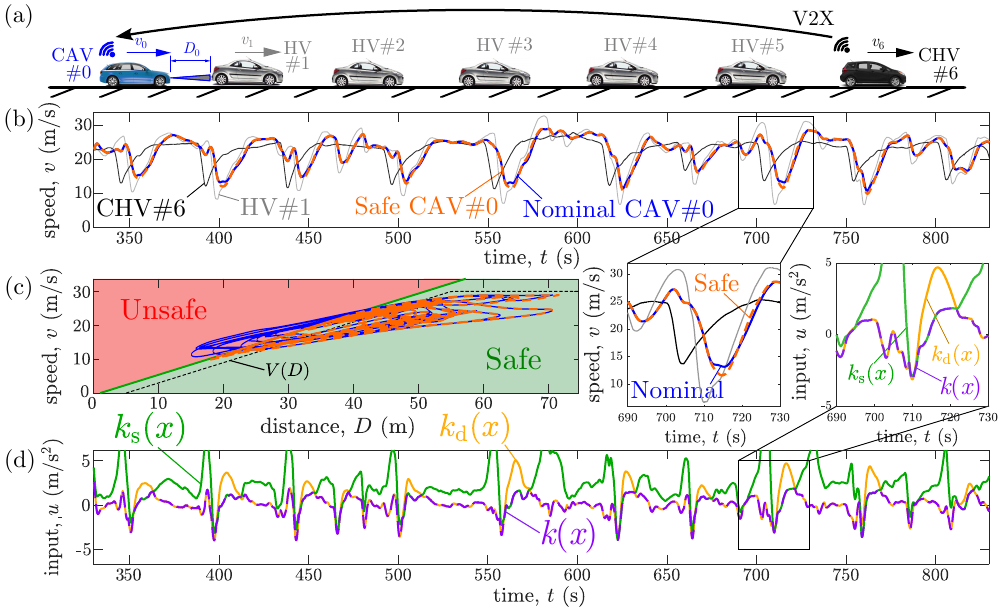}
    \vspace{-3mm}
    \caption{Simulation of a vehicle chain where a CAV (\ref{eq:CAV system w/ lag}) responds to the CHV six vehicles ahead with lag ${\xi=0.6\, \rm{s}}$ using the safety-critical CCC~(\ref{eq:CCC 3veh},\ref{eq:SafCon},\ref{eq:extended SafFilter}). The speed profiles of the CHV and the HVs speed are real experiment data.}
    \label{fig:real data Simulation n=6}
\end{figure*}

\section{Conclusions}
\label{sec:concl}

In this paper, we used control barrier function (CBF) theory to investigate the safety of connected automated vehicles (CAVs) implementing connected cruise control (CCC).
In particular, we developed {\em safety charts} for existing CCC designs to provide safe controller parameters.
We showed that the safe region in these charts shrinks due to the lag in the CAVs dynamics.
The critical value of the lag for which safe regions disappear was also determined.
By comparing safety and stability charts, we demonstrated that CAVs can achieve plant and head-to-tail string stability if they implement safe CCC parameters.
Furthermore, we proposed {\em safety-critical CCC} to achieve high performance while guaranteeing safety despite the lag.
We verified this controller via simulations using real traffic data, and highlighted that connectivity to vehicles farther ahead improves the safety of CAVs during deceleration, and it requires safety filter intervention during acceleration.
In future research, we plan to investigate the nonlinear dynamics of safety-critical CCC.

\newpage
{\color{white}.}
\newpage
{\color{white}.}
\newpage
\section*{Appendix}
\subsection{Expressions of Stability Boundaries}\label{app:A}

In this part, we provide the detailed calculations related to the stability analysis of Section \ref{sec:stability}. 
We consider the case shown in Fig.~\ref{fig:3 veh setup}, where a CAV follows the preceding HV while responding to the CHV that is ${n+1}$ vehicles ahead using CCC (\ref{eq:CCC 3veh}). 
The resulting formulas (\ref{eq:plant stability boundary s=0 final (A,B1)},\ref{eq:plant stability boundary s=jw final (A,B1)},\ref{eq:plant stability boundary s=jw final (B1,Bn)}) and (\ref{eq:string stability boundary w=0 (A,B1)},\ref{eq:string stability boundary w=0 (B1,Bn)},\ref{eq:string stability boundary final (A,B1)},\ref{eq:string stability boundary final (B1,Bn)}) were used in Fig.~\ref{fig: stability boundary analytical}, Fig.~\ref{fig:2veh Safety Chart} and Fig.~\ref{fig:3veh Safety Chart} to plot the plant stability (red) and head-to-tail string stability boundaries (blue), respectively.

For simplicity, we consider that each HV is identical, with ${A_i=A_{\rm h}}$, ${B_i=B_{\rm h}}$, ${\kappa_i=\kappa_{\rm h}}$ and ${V_i(.)=V_{\rm h}(.)}$ in (\ref{eq:HVcon},\ref{eq:HV range policy}) and ${\mathbf{a}_i = \mathbf{a}_{\rm h}}$, ${\mathbf{a}_{i\tau} = \mathbf{a}_{\tau}}$, ${\mathbf{b}_i = \mathbf{b}_{\rm h}}$ and ${\mathbf{b}_{i\tau} = \mathbf{b}_{\tau}}$ in (\ref{eq:linearized dynamics}) for ${i \in \{1,...,n\}}$.
Then, the coefficient matrices in~(\ref{eq:linearized dynamics}) with ${\Phi=\{n+1\}}$ are:
\begin{align}
\begin{split}
\mathbf{a}_0&=\begin{bmatrix}
    0 & -1 & 0
    \\ 
    0 & 0 & 1
    \\ 
    \frac{A\kappa}{\xi} & -\frac{\Psi_0}{\xi} & -\frac{1}{\xi}
    \end{bmatrix}\! \!,
\, \mathbf{b}_{0,1}=\begin{bmatrix}
    1
    \\ 
    0
    \\ 
    \frac{B_1}{\xi}
    \end{bmatrix}\! \!,
\, \mathbf{b}_{0,n+1}=\begin{bmatrix}
    0
    \\ 
    0
    \\ 
    \frac{B_{n\!+\!1}}{\xi}
    \end{bmatrix}\! \!,
    \\
\, \mathbf{a}_{\rm h}&=\begin{bmatrix}
    0 & -1
    \\ 
    0 & 0
    \end{bmatrix}\! \!,
\, \mathbf{a}_{\tau}=\begin{bmatrix}
    0 & 0
    \\ 
    A_{\rm h} \kappa & -\Psi_{\rm h}
    \end{bmatrix}\! \!,
\, \mathbf{b}_{\rm h}=\begin{bmatrix}
    1\\ 
    0
    \end{bmatrix}\! \!,
\, \mathbf{b}_{\tau}=\begin{bmatrix}
    0
    \\ 
    B_{\rm h}
    \end{bmatrix}\! \!,
\end{split}
\label{eq:coefficient matrices}
\end{align}
with ${\Psi_0 = A+B_1+B_{n+1}}$, ${\Psi_{\rm h} = A_{\rm h}+B_{\rm h}}$.
Substituting~(\ref{eq:coefficient matrices}) into (\ref{eq:link transfer function general}) leads to:
\begin{align}
\begin{split}
    T_{0,1}(s) &= \frac{B_1 s + A \kappa}{\xi s^3 + s^2 + \Psi_0 s + A \kappa}, 
    \\
    T_{0,n+1}(s) &= \frac{B_{n+1} s}{\xi s^3 + s^2 + \Psi_0 s + A \kappa}, 
    \\
    T_{i,i+1}(s) &= \frac{B_{\rm h} s + A_{\rm h} \kappa_{\rm h}}{e^{s\tau}s^2 + \Psi_h s + A_{\rm h} \kappa_{\rm h}},
\end{split}
\label{eq:link transfer function detail}
\end{align}
for $i \in \{1,...,n\}$, and (\ref{eq:h2t transfer func 1,n+1}) becomes:
\begin{align}
\begin{split}
    G_{0,n+1}(s) &= T_{0,1}(s)T_{i,i+1}(s)^{n} + T_{0,n+1}(s).
\end{split}
\label{eq:h2t transfer func final}
\end{align}

Since ${\mathrm{D}\big(T_{0,1}(s)\big) = \mathrm{D}\big(T_{0,n+1}(s)\big)}$ given by (\ref{eq:link transfer function detail}), where $\mathrm{D}(.)$ denotes the denominator, the characteristic equation of the system is given by:
\begin{equation}
    \mathrm{D}\big(G_{0,n+1}(s)\big) = \mathrm{D}\big(T_{0,1}(s)\big) \mathrm{D}\big(T_{i,i+1}(s)\big)^n = 0.
\label{eq:characteristic equation detail}
\end{equation}
We assume that all HVs are plant stable, then ${\mathrm{D}\big(T_{i,i+1}(s)\big)=0 \implies {\rm Re}(s)<0}$, and the characteristic equation in (\ref{eq:characteristic equation detail}) can be simplified to:
\begin{equation}
    \mathrm{D}(T_{0,1}(s)) = 0.
\label{eq:characteristic equation simplified}
\end{equation}

Equation (\ref{eq:plant stability boundary, s=0}) gives plant stability boundary when  ${s=0}$. After substituting (\ref{eq:link transfer function detail}) and (\ref{eq:characteristic equation simplified}), it leads to:
\begin{equation}
    A=0.
\label{eq:plant stability boundary s=0 final (A,B1)}
\end{equation}
The ${s=\pm \mathrm{j}\Omega}$ plant stability boundary is obtained by (\ref{eq:plant stability boundary, s=jw}).
By substituting $T_{0,1}(\mathrm{j}\Omega)$ from (\ref{eq:link transfer function detail}) and (\ref{eq:characteristic equation simplified}) into (\ref{eq:plant stability boundary, s=jw}), we get the plant stability boundaries in the ${(A,B_1)}$ space: 
\begin{align}
\begin{split}
    A = \frac{\Omega^2}{\kappa},\quad B_1 = \bigg(\xi - \frac{1}{\kappa} \bigg)\Omega^2 -B_{n+1},
\end{split}
\label{eq:plant stability boundary s=jw final (A,B1)}
\end{align}
or the boundaries in the ${(B_1,B_{n+1})}$ space:
\begin{align}
\begin{split}
    B_{n+1} = A(\kappa\xi -1) - B_1.
\end{split}
\label{eq:plant stability boundary s=jw final (B1,Bn)}
\end{align}

Equations (\ref{eq:P(w)}) and (\ref{eq:string stability boundary, w=0}) provide the ${\omega=0}$ head-to-tail string stability boundary.
After denoting the real and imaginary parts of the HVs' head-to-tail transfer function as ${\GaR = \mathrm{Re}(T_{i,i+1}(\mathrm{j}\omega)^{n})}$, ${\GaI = \mathrm{Im}(T_{i,i+1}(\mathrm{j}\omega)^{n})}$, and substituting (\ref{eq:link transfer function detail}) and (\ref{eq:h2t transfer func final}) into (\ref{eq:P(w)}), we have~\cite{guo2023connected}:
\begin{align}
    P(\omega) &= A^2 \kappa^2 \frac{1-\Gamma_{\rm I}^2(\omega)-\Gamma_{\rm R}^2(\omega)}{\omega^2} - 2 A \kappa \bigg(B_{n+1} \frac{\GaI}{\omega} + 1 \bigg) \nonumber
    \\
    &+ A(A + 2B_1 + 2B_{n+1}) + \mathcal{O}(\omega),
\label{eq:P(w) expand}
\end{align}
where $\mathcal{O}(\omega)$ denotes the higher order terms w.r.t. $\omega$ that vanish as ${\omega \to 0}$. Note that $\mathcal{O}(\omega)$ contains the terms of ${1-\GaR}$ and $\GaI$ since ${\mathrm{lim}_{\omega \rightarrow 0} \GaR = 1}$ and ${\mathrm{lim}_{\omega \rightarrow 0} \GaI = 0}$ hold. 
Then, one can get the ${\omega = 0}$ head-to-tail string stability boundary via (\ref{eq:string stability boundary, w=0}) by taking the limit ${\omega \rightarrow 0}$.
According to L'H$\hat{\rm{o}}$pital's rule, we have \cite{molnar2022virtual}:
\begin{align}
\begin{split}
    \lim_{\omega \rightarrow 0} \frac{\GaI}{\omega} = -\frac{n}{\kappa_{\rm h}},\  
    \lim_{\omega \rightarrow 0} \frac{1-\Gamma_{\rm I}^2(\omega)-\Gamma_{\rm R}^2(\omega)}{\omega^2} = L_{\rm h},
\end{split}
\label{eq:LHopital Rule}
\end{align}
where ${L_{\rm h}=n(A_{\rm h} + 2B_{\rm h} -2\kappa_{\rm h}) / (A_{\rm h} \kappa_{\rm h}^2)}$ and $n$ denotes the number of HVs.
By substituting (\ref{eq:P(w) expand}) and (\ref{eq:LHopital Rule}) into (\ref{eq:string stability boundary, w=0}), we can derive $\omega=0$ string stability boundary:
\begin{align}
\begin{split}
    A = 0,\ A = - \frac{2 \kappa_{\rm h} B_1 + (2n\kappa + 2\kappa_{\rm h}) B_{n+1} - 2 \kappa \kappa_{\rm h}}{\kappa_{\rm h} (L_{\rm h} \kappa^2 + 1)},
\end{split}
\label{eq:string stability boundary w=0 (A,B1)}
\end{align}
in the ${(A,B_1)}$ space, and:
\begin{align}
\begin{split}
    B_1 = - \frac{(L_{\rm h} \kappa^2 + 1)}{2} A - \bigg(\frac{n\kappa }{\kappa_{\rm h}} + 1 \bigg) B_{n+1} + \kappa,
\end{split}
\label{eq:string stability boundary w=0 (B1,Bn)}
\end{align}
in the ${(B_1,B_{n+1})}$ space.

Equation (\ref{eq:string stability boundary, w>0}) defines the ${\omega>0}$ head-to-tail string stability boundary, which can be rearranged into:
\begin{align}
\begin{split}
    d_0(\omega) - n_0(\omega) \mathrm{cos}K + n_1(\omega) \mathrm{sin}K &= 0, 
    \\
    d_1(\omega) - n_0(\omega) \mathrm{sin}K - n_1(\omega) \mathrm{cos}K &= 0.
\end{split}
\label{eq:string stability boundary, w>0 expand}
\end{align}
Substituting (\ref{eq:link transfer function detail}) and (\ref{eq:h2t transfer func final}) into (\ref{eq:string stability boundary, w>0}) and solving it for $A$, $B_1$ and $B_{n+1}$, 
we get the ${\omega > 0}$ head-to-tail string stability boundaries:
\begin{align}
\begin{split}
    A = \frac{{q}_2 {r}_1 - {q}_1 {r}_2}{{p}_1 {q}_2 - {p}_2 {q}_1}, \quad
    B_1 = \frac{{p}_1 {r}_2 - {p}_2 {r}_1}{{p}_1 {q}_2 - {p}_2 {q}_1},
\end{split}
\label{eq:string stability boundary final (A,B1)}
\end{align}
in the ${(A,B_1)}$ space, and:
\begin{align}
\begin{split}
    B_1 = \frac{{q}_4 {r}_3 - {q}_3 {r}_4}{{p}_3 {q}_4 - {p}_4 {q}_3}, \quad
    B_{n+1} = \frac{{p}_3 {r}_4 - {p}_4 {r}_3}{{p}_3 {q}_4 - {p}_4 {q}_3}.
\end{split}
\label{eq:string stability boundary final (B1,Bn)}
\end{align}
in the ${(B_1,B_{n+1})}$ space, with coefficients:
\begin{align}
\begin{split}
    & {p}_1 = \kappa (\GaR \mathrm{cos}K - \GaI \mathrm{sin}K -1), 
    \\  
    & {q}_1 = p_3 = -\omega (\GaR \mathrm{sin}K + \GaI \mathrm{cos}K), 
    \\
    & {p}_2 = p_4 = \kappa (\GaR \mathrm{sin}K + \GaI \mathrm{cos}K) - \omega, 
    \\
    & {q}_2 = \omega ( \GaR \mathrm{cos}K - \GaI \mathrm{sin}K) - \omega,
    \\
    & {r}_1 = \omega (B_{n+1} \mathrm{sin}K - \omega),
    \\
    & {r}_2 = B_{n+1} \omega (1-\mathrm{cos}K) - \omega^3 \xi, 
    \\ 
    & {q}_3 = - \omega \mathrm{sin}K,\quad \quad \ \, {q}_4 = \omega (\mathrm{cos}K - 1),
    \\ 
    & {r}_3 = -p_1 A - \omega^2,\quad \ {r}_4 = -p_2 A - \omega^3\xi.
\end{split}
\label{eq:string stability boundary coefficents (A,B1)}
\end{align}

\newpage
\subsection{Proof of Theorem \ref{theo:Saf Region A-B1-BN}}\label{app:B}

\begin{proof}[Proof of Theorem \ref{theo:Saf Region A-B1-BN}]
We prove safety via Corollary~\ref{cor:extendedNagumo}.
For this, we prove that (\ref{eq:Nagumo Lgh<0}) holds.
We express ${k_{\rm s}(x)-k_{\rm d}(x)}$ by substituting~(\ref{eq:extended SafFilter}) and~(\ref{eq:CCC general}):
\begin{multline}
    k_{\rm s}(x)-k_{\rm d}(x) = \xi \gamma \big( \ksf (v_{1} - v_{0}) - a_{0} \big) +  (1 - \xi \ksf) a_{0} 
    \\
    + \xi \ksf a_{1} + \xi \gamma_{\rm e} \big( \ksf (v_{1} - v_{0}) - a_{0} + \gamma ( \ksf (D_0 - \Dsf) - v_0) \big) 
    \\
    - A \big( V(D_0) - v_0 \big) - B_1 \big( W(v_{1}) - v_0 \big) - \sum_{k \in \Phi} B_k \big( W(v_{k}) - v_0 \big).
    \big).
\label{eq:ks-kd}
\end{multline}
Based on (\ref{eq:V},\ref{eq:W}), we use ${V(D_0) \leq \kappa(D_0-\Dst)}$, ${W(v_{1}) \leq v_{1}}$, ${W(v_{k}) \leq v_{k}}$. Furthermore, we consider ${h_{\rm e}(x)=0}$, ${A \geq 0}$, ${B_1 \geq 0}$, ${B_k \geq 0,\ \forall k \in \Phi}$, which yields:
\begin{multline}
    k_{\rm s}(x) - k_{\rm d}(x) \geq (\xi \gamma + \xi \ksf - 1) \big( \ksf (v_{1} - v_{0}) - a_{0} \big)  
    \\
     + \xi \ksf a_{1} + \big( \ksf (1 - \xi \ksf) - B_1 \big) (v_{1}-v_{0}) 
     \\
     - A \big( \kappa(D_0-\Dst) - v_0 \big) - \sum_{k \in \Phi} B_k \big( v_{k} - v_0 \big).
\label{eq:ks-kd wo nonlinear}
\end{multline}
Then we substitute ${\ksf (v_{1}-v_0) - a_0}$ from ${h_{\rm e}(x)=0}$, i.e., ${\ksf (v_{1}-v_0) - a_0 = -\gamma \big( \ksf(D_0-\Dsf) - v_0 \big)}$, while we add and subtract $Ah(x)$ on the right-hand side, which leads to:
\begin{multline}
    k_{\rm s}(x)-k_{\rm d}(x) \geq A(\kappa_{\rm{sf}} - \kappa)( D_0 - \Dsf) + A \kappa (\Dst - \Dsf) 
    \\
    - (\xi \gamma^2 + \xi \ksf \gamma -\gamma + A)\big(\ksf(D_0 - \Dsf) - v_0 \big) + \xi \ksf a_{1} 
    \\
    + ( \ksf- \xi\ksf^2 - B_1 ) (v_{1} - v_0) - \sum_{k \in \Phi} B_k \big( v_{k} - v_0 \big).
\label{eq:ks-kd wo nonlinear 2}
\end{multline}
Since ${h(x)=\ksf(D_0 - \Dsf) - v_0 \geq 0}$ is considered in (\ref{eq:Nagumo Lgh<0}), the third term on the right-hand side is nonnegative if:
\begin{multline}
    \xi \gamma^2 + \xi \ksf \gamma -\gamma + A = 
    \\
    = \xi \bigg(\gamma - \frac{1-\xi\ksf}{2\xi} \bigg)^2 + A - \frac{(1-\xi\ksf)^2}{4\xi} \leq 0,  
\label{eq:A upper bound}
\end{multline}
which holds given the upper bound on $A$ in (\ref{eq:Saf Region}).
Then we have:
\begin{multline}
\ks(x) - \kd(x) \geq A(\ksf - \kappa)(D_0 - D_{\rm sf})
+ A\kappa(D_{\rm st} - D_{\rm sf} ) 
\\
+ \xi \ksf a_1 - |\ksf - \xi \ksf^2 - B_1| |v_1 - v_0|
- \sum_{k \in \Phi} B_k |v_k - v_0|,
\label{eq:ks-kd wo nonlinear 3}
\end{multline}
where ${D_0 - D_{\rm sf} \geq 0}$ since ${h(x) \geq 0}$ and ${v_0 \geq 0}$.
Finally, considering ${\left | v_1 - v_{0} \right | \leq \bar{v}}$, ${\left | v_k - v_{0} \right | \leq \bar{v}}$, ${\forall k \in \Phi}$, ${a_{1} \geq - a_{\rm{min}}}$, ${\ksf \geq \kappa}$, (\ref{eq:ks-kd wo nonlinear 2}) leads to:
\begin{multline}
    k_{\rm s}(x)-k_{\rm d}(x) \geq A \kappa (\Dst - \Dsf) 
    \\
    - \bigg(\left | \ksf - \xi\ksf^2 - B_1 \right | + \sum_{k \in \Phi} B_k \bigg)\bar{v} - \xi\ksf a_{\rm{min}} \geq 0.  
\label{eq:ks-kd final}
\end{multline}
Considering the lower bound on $A$ in (\ref{eq:Saf Region}) and ${D_{\rm st} > D_{\rm sf}}$, (\ref{eq:ks-kd final}) implies that (\ref{eq:Nagumo Lgh<0}) holds. Therefore, Corollary~\ref{cor:extendedNagumo} provides safety and completes the proof.
\end{proof}

\newpage
\subsection{Proof of Corollary~\ref{cor:critical lag}}\label{app:C}

\begin{proof}[Proof of Corollary \ref{cor:critical lag}]
Based on (\ref{eq:Saf Region}) with ${\gamma \! = \! (1 \!- \! \xi\ksf) \! / \! (2\xi)}$ and ${\xi < 1/\ksf}$,  safe control gains exist only when: 
\begin{equation}
    \frac{(1-\xi \ksf)^2}{4 \xi} \geq \frac{\xi \ksf a_{\rm{min}}}{\kappa (\Dst - \Dsf)}.
    \label{eq:critical lag condition frac}
\end{equation}
Consider:
\begin{equation}
    \rho = \sqrt{\frac{4 \ksf a_{\rm{min}}}{\kappa (\Dst-\Dsf)}},
    \label{eq:rho}
\end{equation}
(\ref{eq:critical lag condition frac}) leads to:
\begin{equation}
    (1 - \xi \ksf + \xi \rho)(1 - \xi \ksf -\xi \rho) \geq 0
    \label{eq:critical lag condition}
\end{equation}
If ${\ksf = \rho}$, (\ref{eq:critical lag condition}) leads to:
\begin{equation}
    \xi \leq \frac{1}{2 \ksf} = \frac{1}{\ksf + \rho}.
    \label{eq:critical lag ksf = rho}
\end{equation}
If ${\ksf > \rho}$, solving (\ref{eq:critical lag condition}) yields:
\begin{equation}
    \xi \leq \frac{1}{\ksf + \rho}\quad \mathrm{or} \quad \xi \geq \frac{1}{\ksf - \rho}.
    \label{eq:critical lag ksf > rho}
\end{equation}
Since ${\xi < 1/\ksf}$, only the former case in (\ref{eq:critical lag ksf > rho}) holds.
If ${\ksf < \rho}$, (\ref{eq:critical lag condition}) gives:
\begin{equation}
    \frac{1}{\ksf - \rho} \leq \xi \leq \frac{1}{\ksf + \rho}.
    \label{eq:critical lag ksf < rho}
\end{equation}
Since ${\xi > 0}$ in Theorem~\ref{theo:Saf Region A-B1-BN}, only the upper bound of the $\xi$ in (\ref{eq:critical lag ksf < rho}) holds. Therefore, ${\xi \leq \xi_{\rm cr} = 1/(\ksf + \rho)}$ holds for all cases, as given by (\ref{eq:critical lag}).
\end{proof}

\subsection{Proof of Theorem 4}\label{app:D}

\begin{proof}[Proof of Theorem \ref{theo:Saf Region general CCC}]
Similar to the proof of Theorem~\ref{theo:Saf Region A-B1-BN}, we prove safety by showing that (\ref{eq:Nagumo Lgh<0}) holds, where
\begin{align}
\begin{split}
    k_{\rm s}(x) &- k_{\rm d}(x) = \xi \gamma \big( \ksf (v_{1} - v_{0}) - a_{0} \big) + (1 - \xi \ksf) a_{0}
    \\
    &+ \xi \gamma_{\rm e} \big( \ksf (v_{1} - v_{0}) - a_{0} + \gamma ( \ksf (D_0 - \Dsf) - v_0) \big) 
    \\
    &+ (\xi \ksf - C_1) a_{1} - A \big( V(D_0) - v_0 \big) - B_1 \big( W(v_{1}) - v_0 \big) 
    \\
    &- \sum_{k \in \Phi} B_k \big( W(v_{k}) - v_0 \big) - \sum_{k \in \Phi} C_k a_{k}.
\end{split}
\label{eq:ks-kd acc}
\end{align}
By following the steps in (\ref{eq:ks-kd wo nonlinear})-(\ref{eq:ks-kd wo nonlinear 3}), the following inequality can be derived:
\begin{align}
\begin{split}
    k_{\rm s}(x) - k_{\rm d}(x) & \geq A(\kappa_{\rm{sf}} - \kappa)( D_0 - \Dsf) + A \kappa (\Dst - \Dsf)
    \\
    &+ (\xi \ksf -C_1) a_{1} - | \ksf- \xi\ksf^2 - B_1 | |v_{1} - v_0| 
    \\
    &- \sum_{k \in \Phi} B_k | v_{k} - v_0 | - \sum_{k \in \Phi} C_k a_{k},
\end{split}
\label{eq:ks-kd wo nonlinear acc}
\end{align}
where the additional acceleration feedback terms show up compared to (\ref{eq:ks-kd wo nonlinear 3}).
Then, considering ${a_{1}, a_{k} \in [-\bar{a}, \bar{a}], \forall k \in \Phi}$, we obtain:
\begin{multline}
    k_{\rm s}(x)-k_{\rm d}(x) \geq A \kappa (\Dst - \Dsf)
    \\
      -\! \big(\left | \ksf \!-\! \xi\ksf^2 \!-\! B_1 \right | \!+\! \sum_{k \in \Phi} B_k \big)\bar{v} \!-\! \big( \left| \xi\ksf \!-\! C_1 \right| \!+\! \sum_{k \in \Phi} |C_k| \big) \bar{a},  
\label{eq:ks-kd final acc}
\end{multline}
cf.~(\ref{eq:ks-kd wo nonlinear 3}).
Using the lower bound on $A$ in (\ref{eq:Saf Region w/ acc}) and ${D_{\rm st} > D_{\rm sf}}$, we obtain that (\ref{eq:Nagumo Lgh<0}) holds and Corollary~\ref{cor:extendedNagumo} implies safety.
\end{proof}

\bibliographystyle{IEEEtran}
\bibliography{references}

\begin{thebibliography}{10}
\providecommand{\url}[1]{#1}
\csname url@samestyle\endcsname
\providecommand{\newblock}{\relax}
\providecommand{\bibinfo}[2]{#2}
\providecommand{\BIBentrySTDinterwordspacing}{\spaceskip=0pt\relax}
\providecommand{\BIBentryALTinterwordstretchfactor}{4}
\providecommand{\BIBentryALTinterwordspacing}{\spaceskip=\fontdimen2\font plus
\BIBentryALTinterwordstretchfactor\fontdimen3\font minus \fontdimen4\font\relax}
\providecommand{\BIBforeignlanguage}[2]{{%
\expandafter\ifx\csname l@#1\endcsname\relax
\typeout{** WARNING: IEEEtran.bst: No hyphenation pattern has been}%
\typeout{** loaded for the language `#1'. Using the pattern for}%
\typeout{** the default language instead.}%
\else
\language=\csname l@#1\endcsname
\fi
#2}}
\providecommand{\BIBdecl}{\relax}
\BIBdecl

\bibitem{Gunter2021arecommercially}
G.~{Gunter}, D.~{Gloudemans}, R.~E. {Stern}, S.~{McQuade}, R.~{Bhadani}, M.~{Bunting}, M.~L. {Delle Monache}, R.~{Lysecky}, B.~{Seibold}, J.~{Sprinkle}, B.~{Piccoli}, and D.~B. {Work}, ``Are commercially implemented adaptive cruise control systems string stable?'' \emph{IEEE Transactions on Intelligent Transportation Systems}, vol.~22, no.~11, pp. 6992--7003, 2021.

\bibitem{Bekiaris-Liberis2018}
N.~Bekiaris-Liberis, C.~Roncoli, and M.~Papageorgiou, ``Predictor-based adaptive cruise control design,'' \emph{IEEE Transactions on Intelligent Transportation Systems}, vol.~19, no.~10, pp. 3181--3195, 2018.

\bibitem{wang2018review_CACC}
Z.~Wang, G.~Wu, and M.~J. Barth, ``A review on cooperative adaptive cruise control ({CACC}) systems: Architectures, controls, and applications,'' in \emph{21st IEEE International Conference on Intelligent Transportation Systems}, 2018, pp. 2884--2891.

\bibitem{Turri2015CACCfuel}
V.~Turri, B.~Besselink, and K.~Johansson, ``Cooperative look-ahead control for fuel-efficient and safe heavy-duty vehicle platooning,'' \emph{IEEE Transactions on Control Systems Technology}, vol.~25, no.~1, pp. 12--28, 2017.

\bibitem{Van2019CACCstability}
E.~van Nunen, J.~Reinders, E.~Semsar-Kazerooni, and N.~van~de Wouw, ``String stable model predictive cooperative adaptive cruise control for heterogeneous platoons,'' \emph{IEEE Transactions on Intelligent Vehicles}, vol.~4, no.~2, pp. 186--196, 2019.

\bibitem{zhang2016motif}
L.~Zhang and G.~Orosz, ``Motif-based design for connected vehicle systems in presence of heterogeneous connectivity structures and time delays,'' \emph{IEEE Transactions on Intelligent Transportation Systems}, vol.~17, no.~6, pp. 1638--1651, 2016.

\bibitem{jin2018experimental}
J.~I. Ge, S.~S. Avedisov, C.~R. He, W.~B. Qin, M.~Sadeghpour, and G.~Orosz, ``Experimental validation of connected automated vehicle design among human-driven vehicles,'' \emph{Transportation Research Part C}, vol.~91, pp. 335--352, 2018.

\bibitem{Qin2018CCCtraffic}
Y.~Qin, H.~Wang, and B.~Ran, ``Control design for stable connected cruise control systems to enhance safety and traffic efficiency,'' \emph{IET Intelligent Transport Systems}, vol.~12, no.~8, pp. 921--930, 2018.

\bibitem{Shen2023Energy}
M.~Shen, C.~R. He, T.~G. Molnar, A.~H. Bell, and G.~Orosz, ``Energy-efficient connected cruise control with lean penetration of connected vehicles,'' \emph{IEEE Transactions on Intelligent Transportation Systems}, vol.~24, no.~4, pp. 4320--4332, 2023.

\bibitem{wang2022LCC}
J.~Wang, Y.~Zheng, C.~Chen, Q.~Xu, and K.~Li, ``Leading cruise control in mixed traffic flow: System modeling, controllability, and string stability,'' \emph{IEEE Transactions on Intelligent Transportation Systems}, vol.~23, no.~8, pp. 12\,861--12\,876, 2022.

\bibitem{TRSC2024}
G.~T. Molnar and G.~Orosz, ``Destroying phantom jams with connectivity and automation: Nonlinear dynamics and control of mixed traffic,'' \emph{Transportation Research Part C: Emerging Technologies}, vol. published online, 2024.

\bibitem{cheng2021enhancing}
Y.~Cheng, C.~Chen, X.~Hu, K.~Chen, Q.~Tang, and Y.~Song, ``Enhancing mixed traffic flow safety via connected and autonomous vehicle trajectory planning with a reinforcement learning approach,'' \emph{Journal of Advanced Transportation}, vol. 2021, pp. 1--11, 2021.

\bibitem{liu2023structural}
D.~Liu, B.~Besselink, S.~Baldi, W.~Yu, and H.~L. Trentelman, ``On structural and safety properties of head-to-tail string stability in mixed platoons,'' \emph{IEEE Transactions on Intelligent Transportation Systems}, vol.~24, no.~6, pp. 6614--6626, 2023.

\bibitem{alam2014guaranteeing}
A.~Alam, A.~Gattami, K.~H. Johansson, and C.~J. Tomlin, ``Guaranteeing safety for heavy duty vehicle platooning: Safe set computations and experimental evaluations,'' \emph{Control Engineering Practice}, vol.~24, pp. 33--41, 2014.

\bibitem{nilsson2015correct}
P.~Nilsson, O.~Hussien, A.~Balkan, Y.~Chen, A.~D. Ames, J.~W. Grizzle, N.~Ozay, H.~Peng, and P.~Tabuada, ``Correct-by-construction adaptive cruise control: Two approaches,'' \emph{IEEE Transactions on Control Systems Technology}, vol.~24, no.~4, pp. 1294--1307, 2015.

\bibitem{Li2020SafeReinforcementLearning}
L.~Wen, J.~Duan, S.~E. Li, S.~Xu, and H.~Peng, ``Safe reinforcement learning for autonomous vehicles through parallel constrained policy optimization,'' in \emph{2020 IEEE 23rd International Conference on Intelligent Transportation Systems (ITSC)}, 2020, pp. 1--7.

\bibitem{massera2017safe}
C.~Massera~Filho, M.~H. Terra, and D.~F. Wolf, ``Safe optimization of highway traffic with robust model predictive control-based cooperative adaptive cruise control,'' \emph{IEEE Transactions on Intelligent Transportation Systems}, vol.~18, no.~11, pp. 3193--3203, 2017.

\bibitem{shen2024energy}
M.~Shen, R.~A. Dollar, T.~G. Molnar, C.~R. He, A.~Vahidi, and G.~Orosz, ``Energy-efficient reactive and predictive connected cruise control,'' \emph{IEEE Transactions on Intelligent Vehicles}, vol.~9, no.~1, pp. 944--957, 2024.

\bibitem{ames2014control}
A.~Ames, J.~Grizzle, and P.~Tabuada, ``Control barrier function based quadratic programs with application to adaptive cruise control,'' in \emph{53rd IEEE Conference on Decision and Control}, 2014, pp. 6271--6278.

\bibitem{Waqas2022ACCCBF}
M.~Waqas, M.~Ali~Murtaza, P.~Nuzzo, and P.~Ioannou, ``Correct-by-construction design of adaptive cruise control with control barrier functions under safety and regulatory constraints,'' in \emph{2022 American Control Conference (ACC)}, 2022, pp. 5140--5146.

\bibitem{chen2018obstacle}
Y.~Chen, H.~Peng, and J.~Grizzle, ``Obstacle avoidance for low-speed autonomous vehicles with barrier function,'' \emph{IEEE Transactions on Control Systems Technology}, vol.~26, no.~1, pp. 194--206, 2018.

\bibitem{hu2023safety}
F.~Hu and H.~Yu, ``Safety-critical lane-change control for {CAV} platoons in mixed autonomy traffic using control barrier functions,'' \emph{arXiv preprint}, no. arXiv:2302.00424, 2023.

\bibitem{abduljabbar2021cbfbased}
M.~Abduljabbar, N.~Meskin, and C.~G. Cassandras, ``Control barrier function-based lateral control of autonomous vehicle for roundabout crossing,'' in \emph{IEEE International Conference on Intelligent Transportation Systems}, 2021, pp. 859--864.

\bibitem{hao2023merge}
V.-A. Le, H.~M. Wang, G.~Orosz, and A.~A. Malikopoulos, ``Coordination for connected automated vehicles at merging roadways in mixed traffic environment,'' in \emph{2023 62nd IEEE Conference on Decision and Control (CDC)}, 2023, pp. 4150--4155.

\bibitem{zhao2023safetycritical}
C.~Zhao, H.~Yu, and T.~G. Molnar, ``Safety-critical traffic control by connected automated vehicles,'' \emph{Transportation Research Part C: Emerging Technologies}, vol. 154, p. {104230}, 2023.

\bibitem{gunter2022experimental}
G.~Gunter, M.~Nice, M.~Bunting, J.~Sprinkle, and D.~B. Work, ``Experimental testing of a control barrier function on an automated vehicle in live multi-lane traffic,'' in \emph{Workshop on Data-Driven and Intelligent Cyber-Physical Systems for Smart Cities}, 2022, pp. 31--35.

\bibitem{He2018}
C.~R. He and G.~Orosz, ``Safety guaranteed connected cruise control,'' in \emph{21st International Conference on Intelligent Transportation Systems}, 2018, pp. 549--554.

\bibitem{Tamas2023CDC}
T.~G. Molnar, G.~Orosz, and A.~D. Ames, ``On the safety of connected cruise control: Analysis and synthesis with control barrier functions,'' in \emph{62nd IEEE Conference on Decision and Control}, 2023, pp. 1106--1111.

\bibitem{alan2023control}
A.~Alan, A.~J. Taylor, C.~R. He, A.~D. Ames, and G.~Orosz, ``Control barrier functions and input-to-state safety with application to automated vehicles,'' \emph{IEEE Transactions on Control Systems Technology}, vol.~31, no.~6, pp. 2744--2759, 2023.

\bibitem{kiss2023safetydelay}
A.~K. Kiss, T.~G. Molnar, A.~D. Ames, and G.~Orosz, ``Control barrier functionals: Safety-critical control for time delay systems,'' \emph{International Journal of Robust and Nonlinear Control}, vol.~33, no.~12, pp. 7282--7309, 2023.

\bibitem{ren2021razumikhin}
W.~Ren, ``Razumikhin-type control {L}yapunov and barrier functions for time-delay systems,'' in \emph{60th IEEE Conference on Decision and Control (CDC)}.\hskip 1em plus 0.5em minus 0.4em\relax IEEE, 2021, pp. 5471--5476.

\bibitem{ren2022razumikhin}
W.~Ren, R.~M. Jungers, and D.~V. Dimarogonas, ``Razumikhin and {K}rasovskii approaches for safe stabilization,'' \emph{Automatica}, vol. 146, p. 110563, 2022.

\bibitem{chen2024CCCsafety}
Y.~Chen, T.~G. Molnar, and G.~Orosz, ``Safety-critical connected cruise control: Leveraging connectivity for safe and efficient longitudinal control of automated vehicles,'' in \emph{27th International Conference on Intelligent Transportation Systems}, 2024.

\bibitem{bando1998analysis}
M.~Bando, K.~Hasebe, K.~Nakanishi, and A.~Nakayama, ``Analysis of optimal velocity model with explicit delay,'' \emph{Physical Review E}, vol.~58, no.~5, p. 5429, 1998.

\bibitem{xiao2022sufficient}
W.~Xiao, C.~A. Belta, and C.~G. Cassandras, ``Sufficient conditions for feasibility of optimal control problems using control barrier functions,'' \emph{Automatica}, vol. 135, p. 109960, 2022.

\bibitem{chen2021backup}
Y.~Chen, M.~Jankovic, M.~Santillo, and A.~D. Ames, ``Backup control barrier functions: Formulation and comparative study,'' in \emph{60th IEEE Conference on Decision and Control}, 2021, pp. 6835--6841.

\bibitem{Ames2021icbf}
A.~D. Ames, G.~Notomista, Y.~Wardi, and M.~Egerstedt, ``Integral control barrier functions for dynamically defined control laws,'' \emph{IEEE Control Systems Letters}, vol.~5, no.~3, pp. 887--892, 2021.

\bibitem{xiao2008Stabilitylag}
L.~Xiao, S.~Darbha, and F.~Gao, ``Stability of string of adaptive cruise control vehicles with parasitic delays and lags,'' in \emph{2008 11th International IEEE Conference on Intelligent Transportation Systems}, 2008, pp. 1101--1106.

\bibitem{guo2023connected}
S.~Guo, G.~Orosz, and T.~G. Molnar, ``Connected cruise and traffic control for pairs of connected automated vehicles,'' \emph{IEEE Transactions on Intelligent Transportation Systems}, vol.~24, no.~11, pp. 12\,648--12\,658, 2023.

\bibitem{molnar2022virtual}
T.~G. Moln{\'a}r, M.~Hopka, D.~Upadhyay, M.~Van~Nieuwstadt, and G.~Orosz, ``Virtual rings on highways: Traffic control by connected automated vehicles,'' in \emph{AI-enabled Technologies for Autonomous and Connected Vehicles}.\hskip 1em plus 0.5em minus 0.4em\relax Springer, 2023, pp. 441--479.

\bibitem{nagumo1942lage}
M.~Nagumo, ``{\"U}ber die lage der integralkurven gew{\"o}hnlicher differentialgleichungen,'' \emph{Proceedings of the Physico-Mathematical Society of Japan. 3rd Series}, vol.~24, pp. 551--559, 1942.

\bibitem{AmesXuGriTab2017}
A.~D. Ames, X.~Xu, J.~W. Grizzle, and P.~Tabuada, ``Control barrier function based quadratic programs for safety critical systems,'' \emph{IEEE Transactions on Automatic Control}, vol.~62, no.~8, pp. 3861--3876, 2017.

\bibitem{Nguyen2016}
Q.~Nguyen and K.~Sreenath, ``{Exponential Control Barrier Functions} for enforcing high relative-degree safety-critical constraints,'' in \emph{American Control Conference}, 2016, pp. 322--328.

\bibitem{xiao2019cbf}
W.~Xiao and C.~Belta, ``Control barrier functions for systems with high relative degree,'' in \emph{58th Conference on Decision and Control}, 2019, pp. 474--479.

\bibitem{cohen2024ROMCBF}
M.~H. Cohen, T.~G. Molnar, and A.~D. Ames, ``Safety-critical control for autonomous systems: Control barrier functions via reduced-order models,'' \emph{Annual Reviews in Control}, vol.~57, p. 100947, 2024.

\bibitem{Ge2014}
J.~I. Ge and G.~Orosz, ``Dynamics of connected vehicle systems with delayed acceleration feedback,'' \emph{Transportation Research Part C}, vol.~46, pp. 46--64, 2014.

\end{thebibliography}

\vspace{-8 mm}
\begin{IEEEbiography}[{\includegraphics[width=1in,height=1.25in,clip,keepaspectratio]{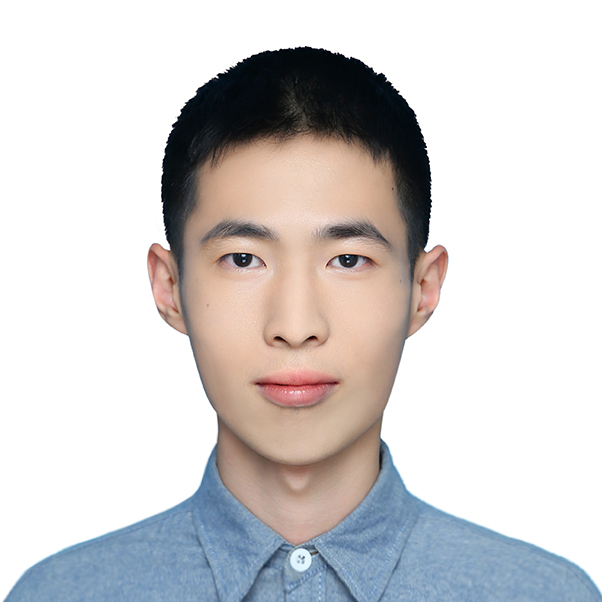}}]{Yuchen~Chen} received the BSc degree in Mechanical Engineering from South China Univerisity of Technology, China, in 2022 and the MSc degree from the University of Michigan, USA, in 2024. He is currently pursuing the PhD degree in mechanical engineering with the University of Michigan, Ann Arbor, USA. His current research interests include control of connected autonomous vehicles, safety-critical control, and vehicle dynamics. 
\end{IEEEbiography}

\vspace{-8mm}
\begin{IEEEbiography}[{\includegraphics[width=1in,height=1.25in,clip,keepaspectratio]{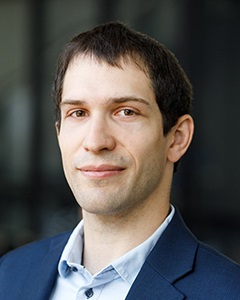}}]
{Tamas G. Molnar} is an Assistant Professor of Mechanical Engineering at the Wichita State University since 2023. Beforehand, he held postdoctoral positions at the California Institute of Technology, from 2020 to 2023, and at the University of Michigan, Ann Arbor, from 2018 to 2020. He received his PhD and MSc in Mechanical Engineering and his BSc in Mechatronics Engineering from the Budapest University of Technology and Economics, Hungary, in 2018, 2015, and 2013. His research interests include nonlinear dynamics and control, safety-critical control, and time delay systems with applications to connected automated vehicles, robotic systems, and autonomous systems.
\end{IEEEbiography}

\vspace{-8mm}
\begin{IEEEbiography}
[{\includegraphics[width=1in,height=1.25in,clip,keepaspectratio]{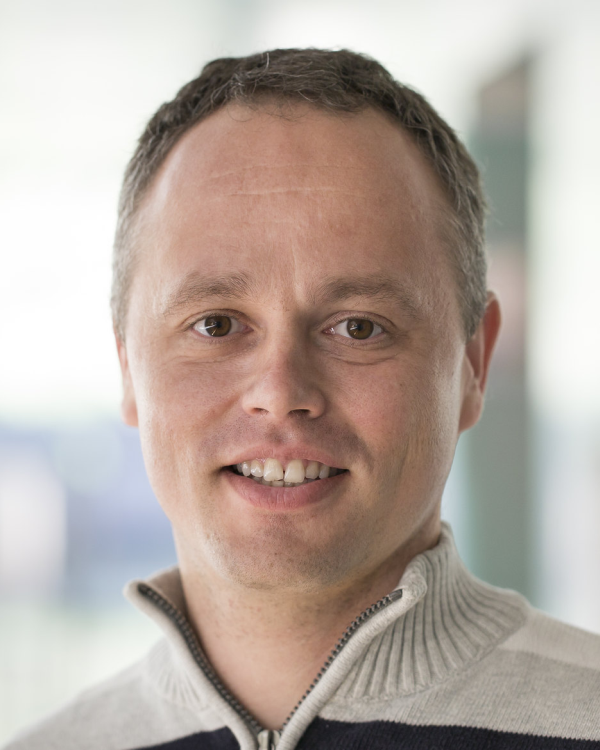}}]{G\'abor Orosz} received the MSc degree in Engineering Physics from the Budapest University of Technology, Hungary, in 2002 and the PhD degree in Engineering Mathematics from the University of Bristol, UK, in 2006. 
He held postdoctoral positions at the University of Exeter, UK, and at the University of California, Santa Barbara. 
In 2010, he joined the University of Michigan, Ann Arbor where he is currently a Professor in Mechanical Engineering and in Civil and Environmental Engineering. 
From 2017 to 2018 he was a Visiting Professor in Control and Dynamical Systems at the California Institute of Technology. 
In 2022 he was a Distinguished Guest Researcher in Applied Mechanics at the Budapest University of Technology and from 2023 to 2024 he was a Fulbright Scholar at the same institution. 
His research interests include nonlinear dynamics and control, time delay systems, machine learning, and data-driven systems with applications to connected and automated vehicles, traffic flow, and biological networks. 
\end{IEEEbiography}

\end{document}